\documentclass[11pt]{article}
\usepackage{amsmath, amssymb, amsthm}
\usepackage{amsopn}
\usepackage{graphicx}
\usepackage{subfigure}
\usepackage[utf8]{inputenc}
\usepackage{nicefrac}
\usepackage{enumerate}
\usepackage[usenames,dvipsnames]{xcolor}
\usepackage[colorlinks=true,linkcolor=Blue,citecolor=Green]{hyperref}
\usepackage{longtable}
\usepackage{emptypage}
\usepackage{bookmark}
\usepackage{fix-cm}
\usepackage[margin=2.5cm,top=2.5cm, bottom=2.5cm]{geometry}

\usepackage{hyperref}
\usepackage{thm-restate}

\usepackage[capitalise]{cleveref}
\crefname{enumi}{}{}

\usepackage{authblk}

\newcommand{\R}{\mathbb R}

\newcommand{\N}{\mathbb N}

\renewcommand{\emptyset}{\varnothing}

\newtheorem{theorem}{Theorem}[section]
\newtheorem{lemma}[theorem]{Lemma}
\newtheorem{proposition}[theorem]{Proposition}

\newtheorem{corollary}[theorem]{Corollary}
\theoremstyle{definition}

\DeclareMathOperator{\cells}{\textit{cells}}
\DeclareMathOperator{\vol}{Vol}

\newcommand{\scalprod}[2]{\left\langle #1,#2 \right\rangle}

\hypersetup{
  pdfauthor={Rafel Jaume and Günter Rote},
  pdftitle={Recursively-regular subdivisions},
  colorlinks=false,
  hidelinks,
}

\title{Recursively-Regular Subdivisions and Applications}
\author{Rafel Jaume \thanks{Research supported by ``Obra Social la Caixa'' and DAAD.}}
\author{Günter Rote \thanks{Partially supported by the ESF programme EuroGIGA-VORONOI, Deutsche Forschungsgemeinschaft (DFG): RO 2338/5-1.}} 
\affil{ \small Institut für Informatik, Freie Universität Berlin. \\ Takustraße 9, D-14195. Berlin, Germany. \\ \texttt{ \{jaume,rote\}@inf.fu-berlin.de} }
\date{}

\begin{document}
 \maketitle
\thispagestyle{empty}
\vspace*{-0.8cm}

\begin{abstract}
We generalize regular subdivisions (polyhedral complexes resulting from the projection of the lower faces of a polyhedron) introducing the class of \emph{recursively-regular subdivisions}. 
Informally speaking, a recursively-regular subdivision is a subdivision that can be obtained by splitting some faces of a regular subdivision by other regular subdivisions (and continue recursively). 
We also define the \emph{finest regular coarsening} and the \emph{regularity tree} of a polyhedral complex. 
We prove that recursively-regular subdivisions are not necessarily connected by flips and that they are acyclic with respect to the in-front relation. 
We show that the finest regular coarsening of a subdivision can be efficiently computed, and that whether a subdivision is recursively regular can be efficiently decided.
As an application, we also extend a theorem known since 1981 on illuminating space by cones and present connections of recursive regularity to tensegrity theory and graph-embedding problems.
\end{abstract}



\section{Introduction}

Regular polyhedral complexes appear in a wide variety of situations. 
The minimization diagram of a set of linear functions, whose regularity follows almost directly from the definition, is a common instance. 
Power diagrams are regular complexes as well. 
It is not hard to see that an arrangement of hyperplanes is a regular subdivision as well; it is the projection of the lower envelope of the dual of a zonotope~\cite{Edelsbrunner}. 
Yet another remarkable example is the Delaunay triangulation of a point set. 
A surprising connection is the Maxwell-Cremona correspondence~\cite{Max-cremona}, which relates the regularity of a planar graph to its rigidity as a framework.

Regular subdivisions are quite well-understood even in higher dimensions. 
Although, as shown by Santos~\cite{Santos2005}, not all the triangulations of a point set in dimension five and higher are connected via flips, regular triangulations are.  
Another remarkable result, which holds in any dimension, is that regular subdivisions contain no cycles in the visibility relations in the sense of~\cite{Edelsbrunner89}. 

On the other hand, not so much is known about non-regular subdivisions. 
Several generalizations of regularity have been studied in order to better understand them. 
For instance, the subdivisions induced by the projection of a polytope onto another polytope, introduced by  Billera, Filliman and Sturmfels \cite{BS}, have been extensively studied together with their variants.

For the clarity of presentation, we will use henceforth the notation $[ n ]$ to refer to the set of natural numbers $\{1,\dots,n \}$. 
The $d$-dimensional Euclidean space will be denoted by $\R^d$ and~$\|\cdot\|$ will denote the Euclidean norm.

\subsection{Polyhedral complexes and subdivisions}
\label{tessandsubd}

Since we will need later several basic results on regular subdivisions, we summarize next the relevant facts and notation. 
See~\cite{LRS} for a detailed discussion on this topic.

We use the term \emph{polytope} for a bounded polyhedron, and \emph{polyhedral cone} refers to the (possibly translated) intersection of finitely many closed linear halfspaces. 
A polyhedral cone is \emph{pointed} if it does not contain any line. 
A \emph{polyhedral complex} is a finite set $\mathcal{S}$ of polyhedra such that if $Q \in \mathcal{S}$ and $F$ is a face of $Q$, then $F \in \mathcal{S}$, and for all $Q,R \in \mathcal{S}$, $Q \cap R$ is a face of both $Q$ and $R$. 
A \emph{polyhedral fan} is a polyhedral complex whose elements are cones. 
A fan is \emph{pointed} if all of its cones are pointed. 
A fan is \emph{complete} if the union of all its cones is the whole ambient space. 
The polyhedra in a polyhedral complex will be called \emph{faces}.
The \emph{dimension} of a polyhedral complex is the dimension of its top-dimensional faces. 
A polyhedral complex is \emph{pure} if all its maximal faces have the same dimension. 
A \emph{cell} is a top-dimensional face of a pure polyhedral complex. 
The set of cells of a polyhedral complex $\mathcal{S}$ will be denoted by $\cells(\mathcal{S})$.
A \emph{facet} or \emph{wall} is a face of co-dimension one in a pure complex. 
As usual, \emph{edges} and \emph{vertices} are one and zero-dimensional faces in the complex, respectively.
The (unbounded) one-dimensional faces of a polyhedral fan will be called \emph{rays} as well.
A pure polyhedral complex embedded in $\R^d$ is \emph{full-dimensional} if it has dimension $d$. 
A $d$-dimensional polyhedral complex is \emph{regular} if its faces are the projection of the lower faces of a $(d+1)$-dimensional polyhedron.  

Let $A \subset \R^d$ be a finite set of points. 
A \emph{polyhedral subdivision} (or \emph{subdivision}, for short) of $A$ is a polyhedral complex whose vertices belong to $A$ and the union of whose cells is the convex hull of $A$. 
A \emph{polyhedral subdivision} (or \emph{subdivision}, for short) of a finite set~$V$ of vectors is a polyhedral fan whose rays have as directions a subset of $V$ and whose union is the positive span of $V$. 
A \emph{triangulation} is a subdivision of a point set consisting only of simplices. 
We use the standard notion of a (geometric bistellar) \emph{flip} between triangulations, see~~\cite[Section~2.4]{LRS} and~\cite{Santos2005}.

Given two complexes $\mathcal{S}$ and $\mathcal{S}'$, we say that $\mathcal{S}'$ \emph{refines} $\mathcal{S}$ if every face of $\mathcal{S}'$ is contained in a face of $\mathcal{S}$. 
We say that $\mathcal{S}'$ is a \emph{refinement} of $\mathcal{S}$ and $\mathcal{S}$ is a \emph{coarsening} of $\mathcal{S}'$. 
The set of subdivisions of a point (or vector) set form a poset with the refinement relation. 
Our notion of a subdivision and the refinement relation are simpler than the more subtle definitions in~\cite[Section 2.3]{LRS} or in~\cite{Santos2005} .
The differences are not relevant to our results.

This article is mainly concerned with recursively-regular subdivisions. 
Intuitively, a polyhedral complex $\mathcal{S}$ is \emph{recursively regular} if it is regular, or it has a regular coarsening $\mathcal{S}'$ such that for each cell $C \in \mathcal{S}'$, the restriction of $\mathcal{S}$ to $C$ is recursively regular. 
Of course, this class of subdivisions generalize regular subdivisions.

For the study of the class of recursively-regular subdivisions, it will be convenient to define the \emph{finest regular coarsening} of a polyhedral subdivision, which is the coarsening of the subdivision that is regular and is not refined by any other regular coarsening. 
The proof for its existence is simple albeit somehow surprising, and it turns out to be an interesting object by itself.

\subsection{Regular subdivisions}
\label{regularsubdivs}
\label{secondary}


Given a point $a \in \R^d$ and a scalar $\lambda \in \R$ we denote by ${a  \choose \lambda} \in \R^{d+1}$ the tuple (thought as a point) resulting from adding the coordinate $\lambda$ to $a$. 
Let $A \subset \R^d$ be a finite set of points.
A subdivision $\mathcal{S}$ of $A$ is \emph{regular} if there exists a \emph{height function} $\omega:  A \rightarrow \R$
such that each face of $\mathcal{S}$ is the projection of a face in the lower convex hull of $$ A^{\omega} = \left\{ {a  \choose \omega(a)} : a \in A \right\}. $$ 
The function $\omega$ will be identified with the vector $\omega \in \R^A$. 
The notation $A^{\omega}$ will be used as a function of a point set $A$ and a height function or vector~$\omega$. 
Given a cell $C \in \cells(\mathcal{S})$, we will also use the notation $$A^\omega |_{C}=\left\{ {a  \choose \omega(a)} : a \in A \cap C \right\}.$$

The following proposition indicates that the regularity of a polyhedral subdivision can be expressed locally. 
We refer to~\cite{LRS} for details. 

\begin{proposition}{(Folding form~\cite{LRS})}\label{foldingform}
Let $A \subset \R^d$ be a finite set of points. 
A polyhedral subdivision $\mathcal{S}$ of $A$ is regular if there exists a height function $\omega:  A \rightarrow \R$ such that for every cell $C \in \cells(\mathcal{S})$, the points of $A^\omega | _{C}$ lie in a hyperplane \emph{(coplanarity condition)}, and for every wall $W=C \cap D$, where $C,D \in  \cells(\mathcal{S})$, the point $a  \choose \omega(a)$ lies strictly above the hyperplane containing  $A^\omega |_{D}$, for all $a \in A \cap ( C \setminus D) $ \emph{(local folding condition)}.
\end{proposition}

In view of the previous result, it is easy to see that the regularity of a subdivision is equivalent to the feasibility of a linear program. 
We sketch the proof of this well-known fact because we will use the notation later. 

Note that the coplanarity condition for a cell can be translated into a set of linear homogeneous equations in the heights of its vertices.
Indeed, it is enough to choose an affine basis for each cell and require that each set resulting from extending this basis with a vertex in the cell is affinely dependent. 
Hence, all the coplanarity conditions together restrict the set of possible height functions $\omega$ to a linear subspace of $\R^{n}$. 

Consider now the local folding condition for a wall $W=C \cap D$ with $C,D \in \cells(\mathcal{S})$.
Let $B=\{b_1, \dots, b_{d+1}\}$ be a spanning set of vertices of $D$, and let $a \in A \cap (D \setminus C) $. 
The local folding condition for $W$ can be expressed~as
\begin{equation}
\label{lfc}
\left|
\begin{array}{cccc}
	1 & \dots & 1  \\
    b_1 & \dots& b_{d+1} \\
\end{array} \right|
 \left|
\begin{array}{cccc}
	1 & \dots & 1 & 1  \\
    b_1 & \dots& b_{d+1} & a\\
    \omega(b_1) & \dots & \omega(b_{d+1}) & \omega(a)\\  
\end{array} \right|
> 0.
\end{equation}
By developing the last row of the second determinant, it becomes clear that this condition is a linear homogeneous strict inequality in the heights of the lifted points. 
Therefore, the local folding conditions for all the walls of a subdivision $\mathcal{S}$ define together a relatively open cone in the subspace determined by the coplanarity conditions. 

\label{regsys}
The \emph{regularity system} of a subdivision is the collection of equations and inequalities resulting from its coplanarity and local folding conditions. 
The \emph{weak regularity system} of a subdivision is the system resulting of replacing the strict inequality in~(\ref{lfc}) with a weak inequality. 
The \emph{secondary cone} is the set of solutions of the weak regularity system. 

Note that the regularity system can be defined for coarsenings of polyhedral complexes, even if they are not polyhedral complexes (that is, if the ``faces'' fail to be convex or the tessellation is not face-to-face). 
Moreover, the definitions and statements presented here can be easily generalized to the case where the initial object $A$ is a set of vectors instead of points. 
In such a case, the cells of the complex are cones forming a polyhedral fan whose $1$-faces are rays with directions taken from $A$. 
The local folding (and coplanarity) conditions lose then the row of ones of both determinants appearing in~(\ref{lfc}) and also one column each, since affine bases must be replaced with linear bases. 
We will use the term subdivision in an ambiguous manner to stress this fact and focus on point-set subdivisions in the proofs. 


\subsection{The secondary fan and the secondary polytope}
Regular subdivisions were first studied by Gelfand, Kapranov and Zelevinsky \cite{GKZ}, who introduced the secondary fan and the secondary polytope. 
These two objects encode the combinatorics of the refinement poset of the regular subdivisions of a point set. 
In addition, they directly imply that regular triangulations are connected by flips---local operations. 
We next give the necessary definitions to state this result.

The GKZ\emph{-vector} $\alpha(\mathcal{T})$ of a triangulation $\mathcal{T}$ of a finite point set $A$ is the vector $\alpha(\mathcal{T}) \in \R^A$ whose $a$-th component is 
\[
\sum \limits_{\substack{C \in \cells(\mathcal{T})  \\ C \ni a } } \vol(C).
\]
The convex hull $\Sigma(A) \subset \R^A$ of all vectors $\alpha(\mathcal{T})$ over all triangulations $\mathcal{T}$ of $A$ is an $(n-d-1)$-dimensional polytope called the \emph{secondary polytope of $A$}.

\begin{theorem}[Gelfand, Kapranov and Zelevinsky~\cite{GKZ}]
The secondary cones of the regular triangulations of a $d$-dimensional point set $A$ define a $(n{-}d{-}1)$-dimensional complete polyhedral fan (called the \emph{secondary fan of} $A$). 
The secondary fan of~$A$ is the normal fan of~$\Sigma(A)$. 
\end{theorem}

As a consequence, the vertices of $\Sigma(A)$ correspond to regu\-lar triangulations of $A$ and the edges of  $\Sigma(A)$ correspond to flips between regular triangulations. 
This proves, in particular, that the regular triangulations of $A$ are connected in the graph of flips.

\subsection{Edelsbrunner's acyclicity theorem}
\label{acycl}
Another nice (geometrically induced) combinatoric property of regular subdivision is the acyclicity in their \emph{in-front} relation.
We state here the definitions we need for later on. 

Let $x$ be a point in $\R^d$ and $S,T \subset \R^d$ be two disjoint convex sets. 
We say that $S$ is \emph{in front of} $T$ with respect to $x$ if there is an open halfline $\ell$ starting at $x$ so that $S_0 = \ell \cap S \neq \emptyset$, $T_0 = \ell \cap T \neq \emptyset$ and every point of $S_0$ lies between $x$ and any point of $T_0$. 
This relation is called the \emph{in-front} relation (from $x$), which is well-defined and antisymmetric because of the convexity of $S$ and $T$. 
The relation can be defined for a direction as well, when $x$ is considered to lie at infinity. 
It can be extended to all the faces of a polyhedral complex. 
To do that, the relation between faces is inherited from the relation between their relative interiors, which are pairwise-disjoint. 
The definition is similarly extended to polyhedral fans. 


A polyhedral complex is said to be \emph{cyclic} in a direction $ v $ (or from a point $x$) if the in-front relation induced by $v$ (or  $x$) on its open cells contains a cycle. 
The complex is called \emph{acyclic} if it is not cyclic from any point or direction. 



\begin{theorem}[Acyclicity Theorem \cite{Edelsbrunner89}]
\label{acyclthm}
Regular polyhedral complexes are acyclic.
\end{theorem}

In fact, Edelsbrunner proved something stronger: that the in-front relation is acyclic for all relatively-open faces of a regular polyhedral complex.

\subsection{Our contribution} 


This article is concerned with recursively-regular subdivisions and some related objects and problems. 
We give a combinatorial characterization of this type of subdivisions based on linear algebra, which leads to efficient algorithms for their recognition and provides meaningful structural properties. 
In~\cref{FRC}, we introduce two constructions closely related to recursively-regular subdivisions: the finest regular coarsening and the regularity tree of a polyhedral subdivision. 
We provide algorithms for the construction of these objects, which have applications in different areas that we explore. 

In addition, we examine some of their combinatorial properties in comparison to regular subdivisions.
In particular, we show that, unlike the regular subdivisions, the recursively-regular subdivisions of a point set are not necessarily connected by bistellar flips. 
On the other hand, recursively-regular subdivisions remain acyclic in the sense of Edelsbrunner (\cref{recareacyclic}).

As the main application, we address the problem of finding (or deciding whether it exists) a one-to-one assignment of a set of floodlights to a set of points such that the floodlights cover the space when translated to the assigned points.
The given floodlights are assumed to be the cells of a complete polyhedral fan.
We say that the fan is \emph{universal} if the floodlights can cover the space regardless of the given point set. 
We prove that recursively-regular fans are indeed universal, and that having a cycle in visibility is sufficient yet not necessary for a fan to be non-universal.
It remains open though to give a characterization of universal fans. 

We also examine two related graph-theoretic problems. 
The first one deals with rigidity of tensegrity frameworks. 
Specifically, we show how to detect the redundant (useless, in a sense) cables from a spider web (tensegrity made of cables and whose convex-hull vertices are pinned). 
The second is concerned with straight line-segment embeddings of digraphs on point sets such that the directions of the arcs satisfy given constraints. 
We show that a big family of digraphs (together with the directions constraints) can be embedded in any given point set whereas some non-trivial digraphs (with constraints) cannot be embedded in some types of point sets. 

\section[The finest regular coarsening and the regularity tree]{The finest regular coarsening and the regularity tree}
\label{FRC}

In this section, we study the finest regular coarsening of a subdivision, which we will use afterwards to define the regularity tree. 
Finally, we will introduce the class of recursively-regular subdivisions and analyze some of its properties.

Roughly speaking, the finest regular coarsening of a subdivision is the finest among all the coarsenings of the subdivision that are regular. 
One should note that it is not obvious whether this object is well-defined. 
We show first that this is indeed the case. 
We do it observing that merging two cells of a subdivision corresponds to converting a local folding condition into a coplanarity condition and, furthermore, this transformation can be done by simply replacing the strict inequality by an equation with the same coefficients. 
In other words, we are looking for the smallest set of inequalities we need to ``relax'' in order to make a given system compatible. 

We first expose a (we assume) well-known fact of linear algebra, for which we could not find a reference. 
We include it for completeness and because it definitely provides an insight into the problem. 
In particular, we give an algorithm in~\cref{mrs} to compute the finest regular coarsening, whose correctness will be implied be the following discussion.

\subsection{A detour through linear algebra}
\label{MRS}
Let $M \in \R^{m \times n}$ be a matrix with row vectors ${s_1,\ldots,s_m \in \R^n}$. 
The \emph{system of~$M$}, denoted by $S(M)$, is the system 
\begin{equation}
\label{system}
S(M): \begin{cases}
\begin{array}{l}
Mx >  0 \\
x \in  \R^n.
\end{array}
\end{cases}
\end{equation}
Given $E \subset [m]$, we use $S^{\geq}(M,E)$ to denote the system 
$$
S^{\geq}(M,E): \begin{cases}
\begin{array}{l}
\scalprod{s_i }{ x}  \geq  0   \text{,  for all } i \in E \\
\scalprod{s_j }{ x}  >  0   \text{,  for all } j  \in  [m] \setminus E \\
x \in  \R^n. 
\end{array}
\end{cases}
$$
Given $E \subset [m]$, the \emph{system of $M$ relaxed by $E$}, denoted by $S(M,E)$, is the system 
$$
S(M,E): \begin{cases}
\begin{array}{l}
\scalprod{s_i }{ x}  =  0   \text{,  for all } i \in E \\
\scalprod{s_j }{ x}  >  0   \text{,  for all } j  \in  [m] \setminus E \\
x \in  \R^n. 
\end{array}
\end{cases}
$$

Literally, the adjective ``relaxed'' would better fit $S^{\geq}(M,E)$ but the following proposition shows that the two systems are equivalent in the cases we are interested in.  
Our purpose is to show that, given a matrix $M \in \R^{m \times n}$, there is a unique set $E \subset [m]$ of minimum cardinality such that $S^{\geq}(M,E)$ has a solution, and that this set can be easily found. 
If $S(M)$ is already compatible, it is clear that $E = \emptyset $ is the unique such set. 
Otherwise, we show that the problem can be transformed into an equivalent one. 

\begin{proposition}
Let $M \in \R^{m \times n}$ be such that $S(M)$ is incompatible, and let $E \subset [m]$ be a set of minimum cardinality such that $S^{\geq}(M,E)$ is compatible. 
Then, $S^{\geq}(M,E)$ and $S(M,E)$ have the same set of solutions.
\end{proposition} 

\begin{proof}
It is clear that the set of solutions of $S^{\geq}(M,E)$ contains the set of solutions of $S(M,E)$. 
Assume $x_0$ is a solution of $S^{\geq}(M,E)$ and is not a solution of $S(M,E)$. 
This means that at least one of the inequalities indexed by $E$ is strictly satisfied by $x_0$. 
If $E_0 \neq \emptyset$ is the set of such inequalities, then $x_0$ is a solution of $S(M, E \setminus E_0)$. 
Since $E_0 \subset E$, this contradicts the assumed minimality of $E$.
\end{proof}

The previous observations motivate the following definition.
Given $M \in \R^{m \times n}$, the \emph{minimum relaxation set} of the system $S(M)$, denoted by $E(M)$, is the intersection of all the sets  $E \subset [m]$ such that $S(M,E)$ is compatible.
The \emph{minimum relaxation} of the system $S(M)$ is the system $S(M,E(M))$.
We will prove that this system is compatible. 
Hence, it will be clear that it is the (unique) set of minimum cardinality that needs to be relaxed in $S(M)$ in order to make the system compatible.  

For our purposes it is easier to argue in terms of the dual problem. 
Given $M \in \R^{m \times n}$, the \emph{dual system of $S(M)$}, denoted by $S^*(M)$, is the system 
\[
S^*(M): \begin{cases}
\begin{array}{l}
 M^\top y =0 \\
 y \geq 0, \; y \neq 0 \\
 y \in \R^m  . 
\end{array}
\end{cases}
\]
A system and its dual are related by the following special case of the Farkas Lemma. 

\begin{lemma}[Gordan's Theorem]
\label{Gordan}
Given $M \in \R^{m \times n}$, the system $S(M)$ is compatible if and only if the dual system $S^*(M)$ is incompatible.
\end{lemma}

This result can be read in the following way: there is no solution for the original system if and only if there exists a non-zero non-negative linear combination $y_0$ of some inequalities leading to the contradiction ``$0>0$''. 
Such a combination $y_0$ is called a\emph{ contradiction cycle} and can be interpreted as a solution to the dual system. 

%

It is convenient to state first the following lemma, which translates Gordan's Theorem to the case where also linear homogeneous equations are included in the system.
Before stating it, we need one more definition. 
Given $M \in \R^{m \times n}$, and $E \subset [m]$, the \emph{dual system of $S(M,E)$}, denoted by $S^*(M,E)$, is the system 
\[
S^*(M,E): \begin{cases}
\begin{array}{l}
 M^\top y =0 \\
 y =(y_1,\dots,y_m)\in \R^m\\
  y_i \geq 0 \,  \text{,  for all } i \in [m] \setminus E  \text{, and}\\
  \text{there exists } \; j \in [m] \setminus E \text{ such that } y_j >0.
\end{array}
\end{cases}
\]


\begin{restatable}[Extension of Gordan's theorem]{lemma}{gordan}
\label{lem:GordanExt}

Given $M \in \R^{m \times n}$, and $E \subset [m]$, the system $S(M,E)$ is compatible if and only if the dual system $S^*(M,E)$ is incompatible.
\end{restatable}


In~\cref{proofGordan}, we prove the previous lemma by reducing it to Gordan's Theorem.
It is also possible to prove it by linear programming duality.

\begin{theorem}
\label{minimalrelaxation} 
Let $M \in \R^{m \times n}$ be  a matrix.
The system $S(M,E(M))$ is compatible.
\end{theorem}

\begin{proof}
If $S(M)$ is compatible, then $E(M)=\emptyset$, since $S(M,\emptyset) = S(M)$ is compatible. 
Assume, then, that $S(M)$ is not compatible.
\cref{Gordan} provides a solution $y_0$ of the dual system $S^*(M)$. 
Let $E_0 \neq \emptyset$ be the set of positive coordinates of $y_0$. 
We will show that $E_0 \subset E$ for any $E \subset [m]$ with $S(M,E)$ compatible. 
Indeed, if we assume the contrary, then $y_0$ is also a solution of $S^*(M,E)$, and we can derive that $S(M,E)$ is not compatible, forcing the contradiction. 
As any set $E$ making $S(M,E)$ compatible must contain $E_0$, we focus now on the system $S(M,E_0)$, which has strictly fewer inequalities than $S$. If it is compatible, then obviously $E(M)=E_0$. 
Otherwise, we keep transforming inequalities into equations iterating the previous arguments (using~\cref{lem:GordanExt}) until a compatible system is found. 
The process finishes because $S(M,[m])$ is compatible. 
Since all the elements we introduce in our relaxation set must be necessarily in any set making the system compatible, it is clear that the set obtained at the end is $E(M)$. 
\end{proof}


An intuitive explanation of why the minimum relaxation is unique can be easily obtained if one looks at the complementary problem. 
That is, given a system of weak homogeneous linear inequalities, decide which is the maximum number of constraints that can be satisfied strictly. 
The set of solutions of the system is a closed polyhedral cone $K$. 
If~$x_0$ is a point in the relative interior of $K$, then the desired maximal set of constraints consists exactly of those constraints that are strictly satisfied by~$x_0$. 
That is, finding this minimum relaxation is equivalent to finding a point in the relative interior of a (possibly not full-dimensional) polyhedral cone given by a set of weak (possibly redundant) inequalities. 

The results proven above can be generalized to systems of non-homogeneous inequalities. 
The main difference would be that $S(M,[m])$ is not necessarily compatible in this case and, therefore, there may be no relaxation at all. 
Nevertheless, whenever there exists a compatible relaxed system, the minimum relaxation is well-defined and can be computed in the same way as in the homogeneous case.

\subsection{The finest regular coarsening of a subdivision}

The algebra developed above will make it very easy to show that there exists a (well-defined) finest regular coarsening of a polyhedral subdivision. 
We next introduce some additional terminology concerning coarsenings. 

Given a polyhedral subdivision $\mathcal{S}$, and a coarsening $\mathcal{S}'$ of $\mathcal{S}$, the \emph{coarsening function} (from $\mathcal{S}$ to $\mathcal{S}'$) is the function $\kappa:\cells(\mathcal{S}) \rightarrow \cells(\mathcal{S}')$ such that $C \subset \kappa(C)$, for all $C \in \cells(\mathcal{S})$. 
Given two coarsenings $\mathcal{S}_1$ and $\mathcal{S}_2$ of $\mathcal{S}$, we say that $\mathcal{S}_1$ is \emph{finer} than $\mathcal{S}_2$ if $\mathcal{S}_2$ is a coarsening of $\mathcal{S}_1$. 
A coarsening is \emph{proper} if it has strictly fewer cells than the original subdivision. 
The \emph{trivial coarsening} is the one merging all the cells into a single one.

Using the definitions in~\cite{LRS}, the refinement relation induces a partial order on the set subdivisions. 
Furthermore, the restriction of this partial order to regular subdivisions is a lattice. 
This lattice is isomorphic to the face lattice of the secondary polytope of the point set.
However, as far as we know, not much work has been done concerning coarsenings of non-regular subdivisions.
The finest regular coarsening goes in that direction, and permits to map every non-regular subdivision to a regular one which is, in a specific sense, the most similar to it.

The \emph{finest regular coarsening} of a subdivision $\mathcal{S}$ of a point set $A$ is the subdivision obtained by the projection of the lower hull of $A^{\omega_0}$, where $\omega_0$ is a solution of the minimum relaxation of the regularity system of $\mathcal{S}$. 
The next theorem justifies the name in the previous definition.

\begin{theorem} 
\label{thm:coarsening}
Let $\mathcal{S}$ be a polyhedral subdivision, and $\mathcal{S}_0$ be the finest regular coarsening of $\mathcal{S}$. 
Then, $\mathcal{S}_0$ is a regular coarsening of $\mathcal{S}$ and all the regular coarsenings of $\mathcal{S}$ are coarsenings of $\mathcal{S}_0$.
\end{theorem} 

\begin{proof} 
Observe first that the relaxation of a constraint corresponding to a local folding condition in the regularity system of $\mathcal{S}$ converts this condition into a coplanarity condition (up to a non-zero scalar factor) for the two cells incident to the wall. 
Thus, the new system is equivalent to the regularity system of the polyhedral complex resulting from merging the two cells of $\mathcal{S}$ involved in the constraint. 
That is, coarsenings of $\mathcal{S}$ have regularity systems that are relaxations of the regularity system of $\mathcal{S}$.
In addition, a coarsening is regular if and only if its regularity system has a solution. 
Hence, $\mathcal{S}_0$ is a coarsening, is regular and it is the regular coarsening that merges the minimum number of cells, that is, the finest one. 
\end{proof}

It will come in handy later to say that a subdivision is \emph{completely non-regular} if its finest regular coarsening is its trivial coarsening. 
This implies, in particular, that every wall of the subdivision can appear in a contradiction cycle of its regularity system.

\subsection{Relation to the secondary polytope}
Note that, once we are convinced that the finest regular coarsening is well-defined for any subdivision of a finite point set $A$, it is easy to derive an alternative definition in terms of the secondary polytope $\Sigma(A)$ of a point set $A$. 
Considering the definitions of subdivision and refinement used in~\cite{LRS}, the faces of $\Sigma(A)$ correspond to regular subdivisions of $A$ and their inclusion relations correspond to coarsening relations. 
The vertices of $\Sigma(A)$ are the GKZ-vectors of all regular triangulations of $A$. 
Non-regular triangulations however have GKZ-vectors that are not vertices of $\Sigma(A)$. 
Moreover, for non-regular triangulations the function mapping a triangulation to its GKZ-vector may not even be injective~\cite{LRS}. 
In any case, the normal cone $\alpha(\mathcal{T})$ of a triangulation $\mathcal{T}$ in $\Sigma(A)$ is isomorphic to the secondary cone of $\mathcal{T}$~\cite{GKZ}. 
It is then not surprising that the finest regular coarsening of a triangulation $\mathcal{T}$ corresponds to the subdivision associated to the smallest face in $\Sigma(A)$ containing $\alpha(\mathcal{T})$. 
As stated in \cite{LRS}, the secondary cone can be similarly defined for general subdivisions (not only for triangulations). 
This cone will be contained in the linear subspace $L$ of the height-functions space defined by the coplanarity conditions. 
Of course, the cone will be also contained in the affine hull $H$ of the secondary fan, which is $(n-d-1)$-dimensional. 
If the dimension of $L \cap H$ is $k < n-d-1$, the subdivision is regular if its secondary cone is $k$-dimensional as well. 
Then, any height function in the relative interior of the secondary cone will certainly produce the subdivision. 
If the subdivision $\mathcal{S}$ is not regular, this cone will not be full-dimensional with respect $L \cap H$. 

\subsection{The regularity tree and recursively-regular subdivisions}
\label{TRP}

Roughly speaking, recursively-regular subdivisions are subdivisions that can be decomposed, via a regular coarsening, into recursively-regular pieces.
More formally, 
a polyhedral subdivision $\mathcal{S}$ is \emph{recursively regular} if it is regular or there exists a proper, non-trivial, and regular coarsening $\mathcal{S}'$ of $\mathcal{S}$ with coarsening function $\kappa$ such that $\kappa^{-1}(C)$ is recursively regular for each cell $C \in \mathcal{S}'$.

Note that the previous definition can be extended to polyhedral fans. 
We will use the notation $\mathfrak{R}(A)$ to refer to the set of recursively-regular subdivisions of a point configuration $A$. 
The class of all recursively-regular subdivisions of any point set will be denoted by $\mathfrak{R}$. 
We will show that $\mathfrak{R}$ is larger than the class of regular subdivisions and that the regularity tree can even have arbitrary depth. 

To proceed, we need to introduce some notation and technical definitions. 
Given a subset $\mathcal{C}$ of $\cells(\mathcal{S})$, we denote by $|  \mathcal{C} | $ the ground set $\cup_{C \in \mathcal{C}}  \, C$ covered by these cells. 
Similarly, if $ \mathcal{S}$ is a subdivision, $|  \mathcal{S} | $ will denote the union of the cells of $\mathcal{S}$. 
A \emph{subdivision tree} of a subdivision $\mathcal{S}$ of a point set $A$ is a rooted tree whose vertices are subsets of $\cells(\mathcal{S})$, whose root is $ \cells( \mathcal{S})$, and such that if the children of $\mathcal{C}$ are $\mathcal{C}_1,\dots,\mathcal{C}_l$, then $|  \mathcal{C}_1| ,\dots,|  \mathcal{C}_l| $ are the cells of a polyhedral subdivision of $A \cap |  \mathcal{C} | $. 
A subdivision tree is called \emph{regular} if the subdivisions of $A \cap |  \mathcal{C} | $ used to split the nodes of the tree are all regular. 

Note that a subdivision is recursively-regular if and only if it has a regular subdivision tree. 
However, a subdivision can have many subdivision trees, and even many regular subdivision trees. 
Fortunately, we can define a canonical one, which will be later used to decide if a subdivision is recursively regular:
The \emph{regularity tree} of the subdivision $\mathcal{S}$ is the subdivision tree created by the following recursion.
\begin{enumerate}[(a)]
\item If a subdivision $\mathcal{S}$ is regular or its finest regular coarsening is trivial, its regularity tree is the tree whose single node is $|\mathcal{S}|$.
\item The regularity tree of a non-regular subdivision $\mathcal{S}$ with a non-trivial finest regular coarsening $\mathcal{S}_0$ is obtained by appending to its trivial coarsening the regularity tree of $\kappa^{-1}(C)$, for each cell $C \in \mathcal{S}_0$.
\end{enumerate}

\cref{Fig:rectree} exhibits an example of a regularity tree. 
The figure shows a triangulation in $\mathfrak{R}$ which needs two levels of recursion to fit the definition of recursively-regular subdivision. 
The coordinates of this example and a proof that the finest regular coarsening of the depicted subdivision is the subdivision defined by the second level of the tree are provided in~\cref{Ap:twolevels}. 
Note that the example consists of a ``pinwheel'' triangulation (refining the ``mother of all examples'' in~\cite{LRS}) inserted into a triangle of a bigger copy of the pinwheel triangulation. 
The insertion procedure can be repeated recursively to obtain a triangulation whose regularity tree has a number of levels linear in the number of vertices.

Note that the leaves of the regularity tree of $\mathcal{S}$ are a partition of $\cells(\mathcal{S})$. 
We say that a \emph{leaf $\mathcal{C}$ is regular}, respectively \emph{completely non-regular}, if the subdivision induced by~$\mathcal{S}$ on $\mathcal{C}$ is regular, respectively completely non-regular.  
By our definition, there are two possibilities for the leaves of the the regularity tree: they are either regular or completely non-regular. 

\begin{figure}[th]
\centering
 \includegraphics[scale=0.7,page=6]{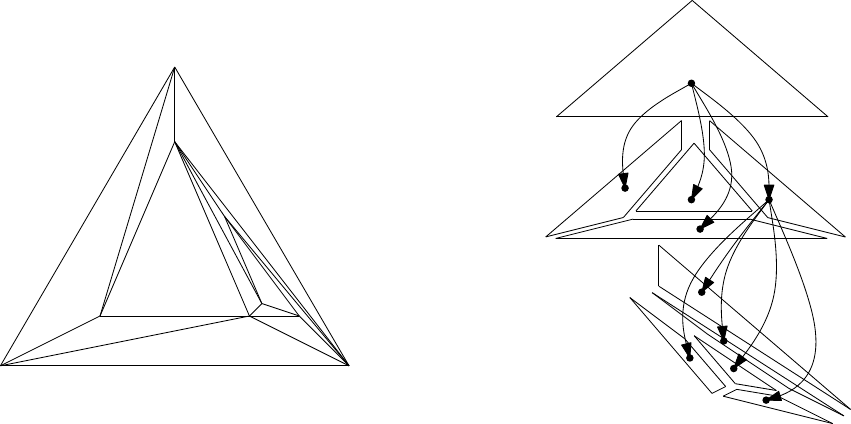}
 \caption{A recursively-regular subdivision and a sketch of its regularity tree.}
 \label{Fig:rectree}
\end{figure}

The following theorem relates the regularity tree and the recursive regularity of a subdivision. 

\begin{theorem}
\label{thm:recursive}
A polyhedral subdivision $\mathcal{S}$ is recursively regular if and only if	 the leaves of its regularity tree are regular.
\end{theorem}

\begin{proof}
If all leaves are regular, the regularity tree itself certifies the recursive regularity of $\mathcal{S}$, proving the \emph{if} direction.

For the \emph{only if}, it will be proved that the leaves of the regularity tree of any subdivision in $\mathfrak{R}$ are regular. 
We do this by induction on the number of cells of the subdivision. 
The base case is when the subdivision consists of a single cell $C$. 
In this case, the only leaf of its regularity tree is $C$, which is regular. 

For the inductive step, let $\mathcal{S}$ be in $\mathfrak{R}$, and assume that the regularity tree of any recursively-regular subdivision with fewer cells than $\mathcal{S}$ has regular leaves. 
Let $\bar{\mathcal{S}}$ be a regular coarsening with coarsening function $\bar{\kappa}$ splitting $\mathcal{S}$ into smaller recursively-regular subdivisions. 
Indeed, by definition, there is a regular subdivision tree of $\mathcal{S}$ representing a set of coarsenings certifying that it is recursively regular. 
We want to show that the regularity tree is a valid certificate as well. 
The second part of \cref{thm:coarsening} asserts that $\bar{\mathcal{S}}$ is a coarsening of the finest regular coarsening $\mathcal{S}_0$ of $\mathcal{S}$. 
This implies that each cell $C \in \mathcal{S}_0$ is contained in some cell $C' \in \bar{\mathcal{S}}$, such that $\mathcal{S}$ restricted to $C'$ is recursively regular. 
Note that refinement relations and regularity behave well with respect to restrictions to polyhedra.
That is, the subdivision obtained by intersecting all the faces of a regular subdivision with a polyhedron is regular as well, and the intersection of a coarsening with a polyhedron is a coarsening of the original subdivision, intersected with the polyhedron. 
Hence, recursive regularity behaves well with respect to restriction to polyhedra and it follows that $\mathcal{S}$ restricted to $C \subset C'$ is recursively regular. 
By induction hypothesis, the leaves of the regularity tree of $\mathcal{S}$ restricted to $C$ are regular, for every $C \in \cells(\mathcal{S}_0)$. 
Since the leaves of the regularity tree of $\mathcal{S}$ are the leaves of the regularity trees of its children, this completes the proof. 
\end{proof}

\begin{figure}[ht]
\begin{center}
\subfigure[Non-regular triangulation]{\includegraphics[scale=0.6,page=1]{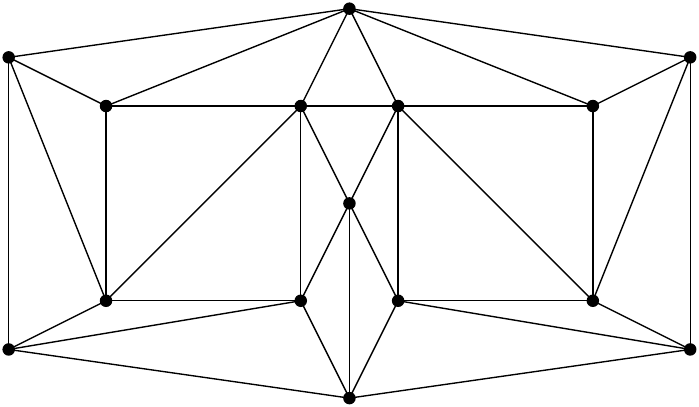} \label{Fig:doublegen}}
\qquad \qquad
\subfigure[Finest regular coarsening]{\includegraphics[scale=0.6,page=2]{./doublegen.pdf} \label{coarsdouble}}
\end{center}
\caption{Two-dimensional recursively-regular and non-regular triangulation.}
\end{figure}

We present now some properties of the recursively-regular subdivisions. 

\begin{proposition}
\label{recareacyclic}
Let $A$ be a finite point set. 
Every regular subdivision of $A$ is recursively-regular. 
Every recursively regular subdivision of $A$ is acyclic.
The converse of the previous statements is in general not true. 
\end{proposition}

\begin{proof}
Note first that regular subdivisions are in $\mathfrak{R}(A)$ by directly applying the definition. 
We will prove that any $\mathcal{S}$ in $\mathfrak{R}$ must be acyclic by induction on its number of cells. 
For the base case, we use that a single-cell subdivision is always acyclic.  
If $\mathcal{S}$ has more than one cell, we distinguish two cases. 
If $\mathcal{S}$ itself is regular, then \cref{acyclthm} shows that it must be acyclic. 
Otherwise, there exists a regular coarsening $\mathcal{S}'$ of $\mathcal{S}$ with coarsening functions $\kappa$. 
Assume for the sake of contradiction that $\mathcal{S}$ contains a cycle and consider the image by $\kappa$ of the involved faces. 
If this image contains more than one cell, the cycle induces another one in $\mathcal{S}'$ leading to a contradiction with its assumed regularity. 
So the cycle must be contained in $\kappa^{-1}(C)$ for a single cell $C \in \mathcal{S}'$. 
But $\kappa^{-1}(C)$ is a recursively-regular subdivision having strictly fewer cells than $\mathcal{S}$ and, hence, acyclic by the induction hypothesis. 

\cref{Fig:doublegen} shows a non-regular triangulation that belongs to $\mathfrak{R}$. 
A certificate for its non-regularity is included in~\cref{Ap:nonreg}, while that it belongs to $\mathfrak{R}$ is straightforward after observing that the coarsening in~\cref{coarsdouble} is regular. 
For the properness of the second inclusion, we refer to the example shown in~\cref{Fig:anrr}, which shows an acyclic subdivision that does not belong to~$\mathfrak{R}$. 
Its acyclicity and that it does not belong to~$\mathfrak{R}$ will be certified in~\cref{app:notrecreg}. 
\end{proof}

\begin{figure}
\begin{center}
\includegraphics[page=1,scale=0.6]{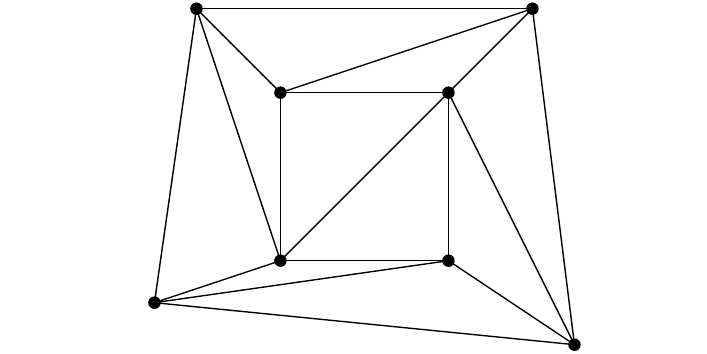} 
\end{center}
\caption{An acyclic triangulation that is not recursively regular.}
\label{Fig:anrr}
\end{figure}


The next proposition illustrates that $\mathfrak{R}$ includes some ``pathological'' triangulations. 
More precisely, we will show that there are triangulations in $\mathfrak{R}$ that are not connected in the graph of flips of its vertex set. 
To prove this, we will simply show that the non-regular triangulations used by Santos in \cite{Santos2005} are indeed in $\mathfrak{R}$.

\begin{proposition}
\label{prop:noflips}
There exists a point set $A \subset \R^5$ whose recursively-regular triangulations are not connected by flips. 
\end{proposition}

\begin{proof}
Santos constructs in~\cite{Santos2005} a set of triangulations $\mathfrak{T}$ of a five-dimensional point set $A$ that are pairwise disconnected in its graph of flips. 
We show that all the triangulations in $\mathfrak{T}$ are recursively regular. 
The convex hull of $A$ is a prism over a polytope $Q$ called the 24-cell. 
The polytope~$Q$ is four-dimensional and has 24 facets, which are regular octahedra. 
All the triangulations in $\mathfrak{T}$ are refinements of the prism $\mathcal{P}$ (in the sense of~\cite[Definition~4.2.10]{LRS}) over a subdivision $\mathcal{B}$ of~$A \cap Q$.
The subdivision $\mathcal{B}$ is a central subdivision of $Q$, it is thus regular (see~\cite[Section~9.5]{LRS}) and consists of 24 pyramids over octahedra. 
Therefore, the prism $\mathcal{P}$ is regular as well (because the prism over a regular subdivision is regular~\cite[Lemma 7.2.4
]{LRS}). 
Each cell of $\mathcal{B}$ is triangulated in a specific way for every triangulation in $\mathfrak{T}$. 
However, since a triangulation of a pyramid is regular if and only if the triangulation induced on its base is regular (see \cite[Observation~4.2.3]{LRS}), and the bases of the pyramids are regular octahedra (which are known to have only regular triangulations), the restriction any triangulation in $\mathfrak{T}$ to any cell of $\mathcal{B}$ is regular. 
Hence, the restriction of any triangulation in $\mathfrak{T}$ to every cell of $\mathcal{P}$ is regular as well, since a triangulation of a prism over a simplex is regular (\cite[Section~6.2]{LRS}).
Thus, every triangulation in $\mathfrak{T}$ is recursively regular. 
Indeed, each triangulation in $\mathfrak{T}$ is a refinement of a regular subdivision $\mathcal{P}$, and its restriction to any cell of $\mathcal{P}$ is regular.
%
\end{proof}

In fact, the previous proposition shows that there is a point set $A$ with at least 12 triangulations in $\mathfrak{R}(A)$ that are pairwise disconnected and disconnected from any regular triangulation in the graph of flips of $A$, as observed in~\cite{Santos2005}. 


\subsection{Algorithms}
\label{alg}
We study how the problem of finding the minimum relaxation of a system, which is equivalent to finding a point in the relative interior of polyhedral cone given by a set of inequalities. 
This problem has been rediscovered several times and an approach to it can be found for instance in~\cite{Fukuda}. 
We give an algorithm starts from a compatible dual system and moves towards compatible primal using the algebraic machinery introduced before.

\begin{proposition}[\textit{folklore}]
\label{mrs}
Let $M \in \R^{m \times n}$.
The minimum relaxation set $E(M)$ of system $S(M)$ (consisting of $m$ linear inequalities on $n$ variables) can be computed solving at most $m$ linear programs in $m$ variables and $n$ constraints.
\end{proposition}

\begin{proof}
In the proof of~\cref{minimalrelaxation} we show that if a coordinate can take a positive value in a solution of $S^*(M)$, then $E(M)$ must include the corresponding index. 
It is also argued that the minimal relaxation of the system can be obtained by incrementally applying this criterion. 
We will convert here this incremental procedure into an algorithm that uses linear programming. 
We start setting $E = \emptyset$ and we insert into the set $E$ the indices that must belong to $E(M)$.
The compatibility of $S(M,E)$ is related to its dual system $S^*(M,E)$. 
To check whether $S^*(M,E)$ has a solution, we solve the linear program 
\[
\textit{maximize} \; \sum \limits_{i \in [m] \setminus E} y_i, 
\] 
subject to the linear constraints given by the system $S^*(M,E)$ plus the condition $\sum_{i \in [m]} y_i \le 1$, which ensures that the maximum is bounded. 
For the ease of argumentation, we add a linear inequality in order to make the dual feasible region bounded. 
If the optimum value is zero, then none of the variables in $[m] \setminus E$ can attain a positive value under the constraints of $S^*(M,E)$, and thus it is incompatible. 
Consequently, the system $S(M,E)$ is compatible and~$E$ is the minimal relaxation (because we have only added an index to $E$ if we know that the index must be in any relaxation set making the system compatible). 
The converse is also true: if the function takes a positive value, a non-empty set of variables $ E_0  \subset [m] \setminus E$ take positive values. 
Therefore, as argued in the proof of~\cref{minimalrelaxation}, we know that $E_0 \subset E(M)$. 
Hence, we add the indices of $E_0$ to $E$ and iterate the process. 
At each iteration, we discover at least one new index that belongs to $E(M)$ and, thus, at most $m$ iterations are needed. 
\end{proof}

%

With help of the previous theorem, it becomes easy to prove that the finest regular coarsening of a subdivision can be efficiently computed. 
The best known bound for linear programming is polynomial only if the total number of bits $L$ needed to encode the coefficients is counted as input size (as in the Turing machine model). 
Alternatively, we can say that a linear program can be solved in time polynomial in the number of variables and $L$. 
We choose this second option to formalize the following bound.  

\begin{corollary} \label{frccomp}
Let $\mathcal{S}$ be subdivision of a point set $A$ in any fixed dimension and let $L$ be the total number of bits necessary to encode the coordinates of $A$.
The finest regular coarsening of~$\mathcal{S}$ can be computed in time polynomial in $|  A | $ and $L$.
\end{corollary}

\begin{proof}
It follows from the definition of the finest regular coarsening that it can be determined by finding a point $\omega_0$ in the relative interior of the secondary cone of  $\mathcal{S}$, computing the point set $A^{\omega_0}$ and its convex hull to finally project its lower faces. 
However, it is easier to iteratively construct it following the algorithm in \cref{mrs} to find the minimum relaxation set of its regularity system. 
Whenever a constraint is relaxed (a dual variable is unrestricted), we merge the cells sharing the corresponding wall. 
We perform the merge operation symbolically by giving a common label to the merged cells. 
When the iteration ends, we construct the cells of the finest regular coarsening by computing the convex hull of the vertices of the cells with the same label.
Since we assume that the dimension is constant and the vertices of the finest regular coarsening are a subset of $A$, the construction of the cells can be done in polynomial time.
 
Note that the coefficients of the linear program come from $d$-dimensional determinants on the coordinates of points in $A$. 
Therefore, the number of bits needed to encode them is polynomial in $L$.
In each iteration, a linear program with a number of constraints proportional to $|  A | $ and as many variables as walls in $\mathcal{S}$ is solved. 
Therefore, the whole algorithm takes polynomial time in $|  A | $ and $L$. 
\end{proof}

Some improvements can probably be done when computing the finest regular coarsening by taking into account the special structure of the regularity system of $\mathcal{S}$. 
In particular, the matrix~$M$ associated to the system is sparse and its structure is related to the combinatorics of the subdivision. 
Each row, corresponding to a wall, has at most $d+2$ non-zero coefficients. 
In addition, $d$ of the involved vertices can be taken to be an affine basis for the corresponding wall. 
Then, the corresponding $d$ coefficients are positive while the other two are negative. 
If $\mathcal{S}$ is a triangulation, this means that each vertex that is involved in a folding condition appearing in a contradiction cycle must be involved in another condition of the contradiction cycle. 
Moreover, if a vertex belongs to the wall corresponding to a condition in a contradiction cycle, it must appear in another condition of the cycle associated to a wall that does not contain it (because the contributions to a vertex in a dual solution must add up to zero). 
If $\mathcal{S}$ is not a triangulation, a similar combinatorial property still holds. 

The statement in the previous corollary is not trivial because there exist subdivisions, even in the plane, with a linear number of simultaneous flips \cite{GHNP}. 
That is, a linear number of pairs of cells that can be independently merged or not. 
Consequently, these subdivisions have an exponential number of minimal coarsenings that one might need to test for regularity. 
The scenario seems even worse when it comes to recursive regularity. 
Fortunately, as a consequence of \cref{thm:recursive}, this can indeed be decided in polynomial time using the procedure in \cref{frccomp}. 

\begin{proposition}
Let $\mathcal{S}$ be subdivision of a point set $A$ in fixed dimension and let $L$ be the total number of bits necessary to encode the coordinates of $A$.
Whether $\mathcal{S}$ is recursively regular can be decided in time polynomial in $|  A | $ and $L$.
\end{proposition}

\begin{proof}
\cref{thm:recursive} ensures that we only need to compute the regularity tree of $\mathcal{S}$ to decide whether $\mathcal{S}$ belongs to $\mathfrak{R}$ or not. This is done by computing the finest regular coarsening of subdivisions of some subsets of $A$. Each time we go down a level in the tree, there is one wall in the finest regular coarsening that was not in any previous finest regular coarsenings. 
Therefore, if we charge the computation of the finest regular coarsening to this wall, we can conclude that the number of computations is bounded by the number of walls in $\mathcal{S}$, which is polynomial if $d$ is considered to be a constant.
\end{proof}

\section[Illumination by floodlights in high dimensions]{Illumination by floodlights in high dimensions}
\label{floodlights}

In the last decades, a wide collection of problems have been studied concerning illumination or guarding of geometric objects. 
The first Art Gallery problem posed by Klee asked simply how many guards are necessary to guard a polygon. 
Since then, considerable research has addressed several variants of this problem, such as finding watchman routes or illuminating sets of objects. 
A remarkable group of problems arises when the light sources (or the surveillance devices) do not behave in the same way in all directions. 
In the major part of the literature, these problems are studied only in the plane. 
A compilation of results on this type of problem can be found in \cite{Urrutia}. 
The problem we are interested in assumes that a light source can illuminate only a convex unbounded polyhedral cone. 
We are given the polyhedral cones available and a set of points representing the allowed positions for their apices. 
We can then choose the assignment of the floodlights to the points in order to cover some target set. 
The assignment will be required to be one-to-one and the floodlights will not be permitted to rotate. 

The first problem we look at in this section is the \emph{space illumination problem} in three or higher dimensions. 
Informally speaking, the problem asks if given a set of floodlights and a set of points there is a placement of the floodlights on the points such that the whole space is illuminated.
Afterwards, we study the generalization to higher dimensions of the \emph{stage illumination problem}, introduced by Bose, Guibas, Lubiw, Overmars, Souvaine and Urrutia \cite{BGAMSU}. 
First of all, we reproduce a result that will be used in the subsequent proofs. 

\subsection{Power diagrams and constrained least-squares assignments}
\label{sec:power}
We present here a connection between least-squares assignments and power diagrams observed in~\cite{AHA}.
%
%

The \emph{power diagram} of a finite set of points $Q \subset \R^d$ (called \emph{sites}) with assigned weights $w: Q \to \R^+$ is the polyhedral complex whose cells are 
$$
R_q=\{x \in \R^d: \| x-q \|^2 -w(d) \leq \| x-q' \|^2 -w(q') \mbox{ for all } q' \in Q  \},\text{ for } q \in Q.
$$
For every  $q\in Q$, the locus $R_q$ is a polyhedron called the \emph{region} of $q$. 
Note that, for every $q\in Q$, the value $\| x-q \|^2 -w(d)$ is the power of the point $x$ with respect to a circle centered at $q$ and having radius $\sqrt{w(q)}$, in case $w(q) \ge 0$. 
This is the reason why the power diagram is often defined for a set of circles instead of weighted points. 
For more details on this type of diagrams, see the survey in~\cite{AurenhammerPower}. 


Given a finite point set $S$ and a set $Q$ of weighted points, we say that an assignment $\sigma: S \to Q$ is \emph{induced by the power diagram} of $Q$ if $ \sigma^{-1}(q) \subset R_q$, for all $q \in Q$. 

Given a finite set of points $Q$, a function $c:Q \to \N$, and a set $S$ of $\sum_{q \in Q}  c(q)$ points, a \emph{constrained least-squares assignment} for $Q$ and $S$ with \emph{capacities} $c$ is an assignment minimizing 
$ \sum_{q \in Q} \| b- \tau(b)  \|^2 $
among all $\tau: S \to Q$ satisfying  $|\tau^{-1}(q)|=c(q)$ for all $q \in Q$. 


\begin{theorem}[Aurenhammer, Hoffmann and Aronov~\cite{AHA}]
Let $Q$ be a finite set of points with weights $w: Q \to \R^+$ and let $S$ be a point set. 
\begin{enumerate}[(i)]
\item Any assignment $\sigma:S \to Q$ induced by the power diagram of $Q$ is a constrained least-squares assignment for $Q$ and $S$ with capacities $c(q)=|\sigma^{-1}(q)|$, for all $q \in Q$. 
\item Conversely, if $\pi$ is a constrained least-squares assignment for $Q$ and $S$ with capacities $c$, then there exist $w$ such that $\pi$ is induced by the power diagram of $Q$ weighted by $w$. 
\end{enumerate}
\end{theorem}

The analogous result can be stated replacing $S$ by a continuous measure. 
In particular, the measure could be uniform in a polytope and the capacities would then be a partition of its volume. 
As a consequence, the following Minkowsky-type theorem can be easily derived. 

\begin{theorem}[Aurenhammer, Hoffmann and Aronov~\cite{AHA}]\label{prop:partition} 
Let $Q$ be a finite set of points. 
\begin{enumerate}[(i)]
\item For any finite set $S$ of points and function $c:Q \to \N $ such that $\sum_{q \in Q} c(q)= |S|$, there exist weights $w: Q \to \R^+$ such that the power diagram of $Q$ weighted by $w$ induces an assignment $\sigma: S \to Q$ with $|\sigma^{-1}(q)|=c(q)$, for all $q \in Q$.\label{prop:dispow} 
\item For any polytope $P$ and function $c:Q \to \R^+$ such that $\sum_{q \in Q} c(q)= \vol(P)$, there exist weights $w: Q \to \R^+$ such that the power diagram of $Q$ weighted by $w$ induces an assignment $\sigma: P \to Q$ with $\vol(\sigma^{-1}(q))=c(q)$, for all $q \in Q$.\label{prop:contpow}
\end{enumerate}
\end{theorem}


If $k=|Q| \leq |S|=n$, a partition as indicated in~\cref{prop:partition}-(\ref{prop:dispow}) can be computed in $O(k^2 n \log n )$ time by an algorithm given in~\cite{Alberts}. 
For the special case $|S|=|Q|=n$ and $c(q)=1$ for all $q \in Q$, the problem of finding such a partition can be formulated as a linear sum assignment problem and can be thus solved using the Hungarian method in $O(n^3)$ time.

\subsection{Illuminating space}\label{illumin}

The results presented here use recursively-regular polyhedral fans. 
These objects are analogous to recursively-regular subdivisions of a point set with vectors instead of points as base elements.
We next introduce some new definitions specific to this problem. 
The \emph{ground set} of a polyhedral fan $\mathcal{F}$, denoted by $|  \mathcal{F} | $, is the union of all its cells. 
We say that a $d$-dimensional polyhedral fan is \emph{complete} if its ground set is the whole space and that it is \emph{conic} if the ground set is a pointed $d$-dimensional cone. 
Similarly, we will talk about the \emph{complete case} and the \emph{conic case} to refer to instances of the problem where the given fan is complete or conic, respectively. 
A facet of a fan will be called \emph{interior} if it is not contained in the boundary of the ground set of the fan. 
A cone $K$ is said to \emph{contain a direction (or vector) $v$} if it contains the ray $\rho_v$ starting at the apex of $K$ and having direction (or direction vector) $v$. 
We will say that the direction is \emph{interior to a cone} if  $\rho_v$ intersects the boundary of $K$ only in its apex. 

Let $P$ be a polyhedron 
\[
P= \bigcap_{i \in I} \Pi_i^+,
\]
where $\Pi_i$ are the hyperplanes supporting the facets of $P$, for $i \in I$.
The \emph{reverse polyhedron} of~$P$, denoted by $P^-$, is defined as
\[
P^-= \bigcap_{i \in I} \Pi_i^-.
\]
The \emph{reverse fan} of a polyhedral fan $\mathcal{F}$ is the fan obtained by reversing all its faces. The \emph{reverse cone} of a conic fan is the reversed set of its ground set.  
Note that if $P$ is a cone with apex at the origin, then $P^-=-P$.


Given a $d$-dimensional complete polyhedral fan $\mathcal{F}$ with $n$ cells and a set of $n$ points $P \subset \R^d$, we say that an assignment $\sigma: \cells(\mathcal{F}) \rightarrow P$ is \emph{covering} if it is one-to-one and $$\bigcup \limits_{C \in \cells(\mathcal{F})} \left( C+{\sigma(C)} \right) \supset |  \mathcal{F} | .$$
Note that the floodlights are only translated to the corresponding points and not rotated, as in other variants of the problem. 


We are now ready to state formally the \emph{space illumination problem}. 
Given a $d$-dimensional polyhedral fan and a set of points in $\R^d$ we would like to know whether there is a covering assignment for that fan and the point set. 
Galperin and Galperin~\cite{Galperin} proved that a covering assignment can be found if the fan is complete and regular, regardless of the given point set and in any dimension.

\begin{theorem}[Galperin and Galperin~\cite{Galperin}; Rote~\cite{RoteCG}]
\label{thm:Gal}
Let $\mathcal{F} \subset \R^d$ be a full-dimensional regular polyhedral fan consisting of $n$ cells and $P \subset \R^d$ be a set of $n$ points. 
There is a covering assignment for $\mathcal{F}$ and $P$. 
\end{theorem}

In particular, the there is a covering assignment for a fan in the plane and any point set of the right cardinality. 
This last statement was rediscovered with a small variation in the formulation of the problem in \cite{BGAMSU}, where an $O(n \log n)$ algorithm for finding a covering assignment is given as well. 
The conic case in the plane has also been considered with the extra assumption that the points are contained in the reverse cone of the fan. 
In this case, a covering assignment can be always found as well. 
However, if the points are not required to lie in the reverse cone, deciding the existence of a covering assignment becomes NP-hard even in the plane, since the problem is equivalent to the \emph{wedge illumination problem} studied in \cite{CRSV}. 
It is worth mentioning the problem of illumination disks with a minimum number of points in the plane, studied by Fejes T{\'o}th~\cite{Toth}. 
The notion of illumination in that work is not the usual one and the lights can be placed anywhere. 
Surprisingly enough, he used the properties of power diagrams to prove an upper bound on the number of needed points, the same diagrams used by Rote~\cite{RoteCG} to provide an alternative proof of~\cref{thm:Gal}.

We generalize first the conic case to higher dimensions and prove that it is sufficient for the fan to be recursively-regular to ensure the existence of a covering assignment for any point set in the reverse cone of the fan. 
Afterwards, we use this result to prove that~\cref{thm:Gal} can be extended to recursively-regular fans in the complete case as well. 
Both generalizations are synthesized in the following statement (note that $(\R^d)^-=\R^d$ and there is thus no restriction for $P$ in the complete case). 

\begin{theorem}
\label{thm:recreg}
Let $\mathcal{F} \subset \R^d$ be a full-dimensional recursively-regular polyhedral fan consisting of $n$ cells and $P \subset |  \mathcal{F} | ^-$ be a set of $n$ points. There is a covering assignment for $\mathcal{F}$ and $P$. 
\end{theorem}

We prove first two simple technical lemmas. 

\begin{lemma}\label{Lemma:conic}
A conic full-dimensional fan $\mathcal{F} \subset \R^d$ with $|  \mathcal{F} |  =K$ is regular if and only if $\mathcal{F}$ is the restriction to $K$ of a complete regular fan.
\end{lemma}

\begin{proof}
For the \emph{only if} direction, assume that $\mathcal{F}$ is regular and, hence, there is a cone $\widetilde{K}\subset \R^{d+1}$ whose lower convex hull projects on $\mathcal{F}$. 
This cone can be written as $$\widetilde{K}=\left(\bigcap_{i \in I^+} \Pi_i^+\right) \cap \left(\bigcap_{i \in I^-} \Pi_i^-\right),$$ where $\Pi^+$ refers to the closed halfspace above the hyperplane $\Pi$ and $\Pi^-$ refers to the closed halfspace below $\Pi$; and $I^+$ is the set of indices such that $\widetilde{K} \subset \Pi_i^+$ and $I^-$ is the set of indices such that $\widetilde{K} \subset \Pi_i^-$. 
By convention, $\widetilde{K}$ will be considered to lie below the vertical hyperplanes, and thus the indices of these hyperplanes are considered as part of $I^-$. 
Note that $\bigcap_{i \in I^+} \Pi_i^+$ is a cone whose faces project onto a complete fan $\mathcal{G}$, since the vertical direction is interior to it. 
Moreover, its restriction to $|  \mathcal{F} | $ is $\mathcal{F}$. 

To prove the \emph{if} direction, assume that $\widetilde{L} \subset \R^{d+1}$ is a cone whose lower hull projects onto a complete fan $\mathcal{G} \subset \R^d$ and let $K= \bigcap_{i \in I} \Pi_i^+ \subset \R^d$ be a polyhedral cone. 
For every $i \in I$, let $\widetilde{\Pi}_i$ be the vertical hyperplane in $\R^{d+1}$ containing $\Pi_i$.
Clearly the set $\widetilde{L} \cap (\bigcap_{i \in I} \widetilde{\Pi}_i^+)$ is a cone whose lower hull projects onto the restriction of $\mathcal{G}$ to~$K$.
\end{proof}

The following technical lemma will be useful to extend~\cref{thm:Gal} to the conic case and to recursively-regular fans. 
Given a complete fan $\mathcal{G}$ and a pointed cone $K$, the lemma relates the existence of a covering assignment for $\mathcal{G}$ to a covering property of the restriction $\mathcal{F}$ of $\mathcal{G}$ to $K$.	
More precisely, we show that the cells of $\mathcal{F}$ can cover a polyhedron $Q$ resulting from shifting the hyperplanes defining $K$ provided that the given point set lies in $Q^-$ and that there is a covering assignment for this point set and $\mathcal{G}$.  

%
%
%
%

\begin{lemma}\label{Lemma:reverse}
Let $Q=\bigcap_{i \in I} (\Pi_i^+ + t_i)$ be a full-dimensional polyhedron, where $t_i \in \R^d$ for all $i \in I$.
Let $\mathcal{G} \subset \R^d$ be a full-dimensional complete fan consisting of $n$ cells, whose restriction $\mathcal{F}$ to $K = \bigcap_{i \in I} \Pi_i^+$ consists of $n$ cells as well. 
If there is a covering assignment for $\mathcal{G} $ and a set $P \subset Q^-$  of $n$ points, then the cells of $\mathcal{F}$ translated by the corresponding assignment cover $Q$.
\end{lemma}

\begin{proof}
Let $\theta: \cells(\mathcal{F}) \to \cells(\mathcal{G})$ be the map such that $C = \theta(C) \cap K $ for all $C \in \cells(\mathcal{F})$, and let $\sigma: \cells(\mathcal{G}) \to P$ be a covering assignment. 
We want to show that 
\[
\bigcup \limits_{C \in \cells(\mathcal{F})} \left( C+{\sigma(C)} \right) \supset Q.
\]
By hypothesis,
\[
\bigcup \limits_{C \in \cells(\mathcal{F})} \left( \theta(C)+{\sigma(C)} \right)=\R^d \supset Q.
\]
We are done if we can prove that $ (C+p) \cap Q \supset (\theta(C)+p) \cap Q$ for all $p \in Q^-$ and for all~$C \in \cells(\mathcal{F})$.
Note that $Q \subset \Pi_i^+ +p$ for any $p \in  Q^-$ and for all $i \in I$, by definition of reverse polyhedron. 
Since \[C = \theta(C) \cap K =\theta(C) \cap \left( \bigcap_{i \in I} \Pi_i^+\right),\] it follows that 
\[ (C+p) \cap Q =  \left[ (\theta(C)+p) \cap \left( \bigcap_{i \in I} (\Pi_i^+ +p) \right)  \right] \cap Q \supset (\theta(C)+p) \cap Q .\]
Therefore, the cells of $\mathcal{F}$ translated according to $\sigma \circ \theta$ cover $Q$.
\end{proof}

The following proposition is now easy to prove.

\begin{proposition}\label{Theorem:Main}
Let $\mathcal{F} \subset \R^d$ be a full-dimensional conic regular fan with $| \mathcal{F}| =K$ consisting of $n$ cells and $P \subset K^-$ be a set of $n$ points. 
There is a covering assignment for $\mathcal{F}$ and $P$.
\end{proposition}

\begin{proof} 
\cref{Lemma:conic} provides us with a fan $\mathcal{G}$ whose restriction to $K$ coincides with $\mathcal{F}$ and has the same number of cells. 
\cref{thm:Gal} applies to $\mathcal{G}$ and $P$. 
It only remains to invoke~\cref{Lemma:reverse} to show that any covering assignment for $\mathcal{G}$ and $P$ can trivially be translated into a covering assignment for $\mathcal{F}$ and $P$.
\end{proof}

We can now prove \cref{thm:recreg}, which generalizes the result (and the proof) in \cite{RoteCG}. 

\begin{proof}[Proof of \cref{thm:recreg}]
The proof proceeds recursively splitting the set of cells of $\mathcal{F}$, the space in $|\mathcal{F}|$ and the points of $P$ into smaller problems. 
Before detailing the recursion, we introduce some notation and include a proof of a lemma by Rote. 

Let $\mathcal{F}_0$ be the finest regular coarsening of $\mathcal{F}$, and $\kappa: \cells(\mathcal{F}) \to \cells(\mathcal{F}_0)$ be the associated coarsening function.
Let $$K= \underset{C \in \cells(\mathcal{F}_0)} \bigcap \Pi_C^+$$ be a $(d{+}1)$-dimensional cone projecting onto $\mathcal{F}_0$, where the hyperplane $\Pi_C$ supports the facet of $K$ that projects onto $C$, for all $C \in \cells(\mathcal{F}_0)$.  
Given a function $\omega:\cells(\mathcal{F}_0) \to \R$, let the power diagram $\varphi^*(\mathcal{F}_0,\omega)$ be the (projection of the) upper envelope of the hyperplane arrangement obtained by vertically shifting the hyperplane $\Pi_C$ by $\omega(C)$, for all $C \in \cells(\mathcal{F}_0)$. 
Similarly, let $\varphi_*(\mathcal{F}_0,\omega)$ denote the lower envelope of these hyperplanes. 
Both power diagrams have as many cells as $\mathcal{F}_0$ and all of them are unbounded. 
In addition, the cells in these diagrams can be paired in a natural way with the hyperplane they come from. 
For simplicity of notation, let $C^*$ denote the cell of $\varphi^*(\mathcal{F}_0,\omega)$ corresponding to $C$, and by $C_*$ the corresponding cell of $\varphi_*(\mathcal{F}_0,\omega)$. 
These pairs of cells satisfy the following property, which is illustrated in \cref{fig:powerdual}. 

\begin{lemma}[Rote \cite{RoteCG}]\label{lem:powerdual}
Every cell $C_* \in \varphi_*(\mathcal{F},\omega)$ is contained in the reverse polyhedron of $C^* \in \varphi(\mathcal{F},\omega)$.
\end{lemma}

\begin{proof}
Choose an arbitrary cell $C_*$ of $\varphi_*(\mathcal{F},\omega)$.
Let $D_*$ be an adjacent cell and $W_*=C_* \cap D_*$ be their common wall. 
Consider also the wall $W^*=C^* \cap D^*$. 
Note that both $W_*$ and $W^*$ are supported by the hyperplane $h$, which is the projection of $(\Pi_C+\omega(C)) \cap (\Pi_D+\omega(D))$. 
Clearly $\Pi_C+\omega(C)$ is above $\Pi_D+w_D$ in one side of $h$ while $\Pi_D+\omega(D)$ is above $\Pi_C+\omega(C)$ in the other side and, hence, $C_*$ is contained in one side of $h$ while $D_*$ is contained in the other.
Putting together the analogous observations for all other cells adjacent to $C_*$, we derive the desired statement for this (arbitrarily chosen) cell~$C$. 
\end{proof}

We continue the proof of~\cref{thm:recreg}. 
The main idea is to apply \cref{prop:partition}-(\ref{prop:dispow}) and find~$\omega$ such that $\varphi_*(\mathcal{F}_0,\omega)$ leaves in $C_*$ exactly $|\kappa^{-1}(C)|$ points of~$P$, for every cell $C\in \cells(\mathcal{F}_0)$. 
Then, we will cover each cell $C^*$ of $\varphi^*(\mathcal{F}_0,\omega)$ with the floodlights of $\mathcal{F}$ contained in~$C$ and the points of $P$ in $C_*$. 
If $\mathcal{F}_C=\kappa^{-1}(C) \cap \mathcal{F}$ is regular, we proceed as in the proof of~\cref{Theorem:Main}, using~\cref{Lemma:reverse} to construct an assignment that covers $C^*$ with the points in $P \cap C_*$ and the floodlights of $\mathcal{F} \cap C$. 
If $\mathcal{F}_C=\kappa^{-1}(C) \cap \mathcal{F}$ is not regular but recursively-regular, we repeat the process recursively. 
That is, we split the points of $P \cap C^*$ with a power diagram associated to the finest regular coarsening $\mathcal{G}_0$ of $\mathcal{F}_C$. 
For each cell $D$ of $\mathcal{G}_0$, we get points in $C_* \cap D_*$, which is contained in the reverse polyhedron of $C^* \cap D^*$. 
Hence, the recursion proceeds until the base case, where we can cover the target polyhedron with a regular fan from points in its reverse polyhedron using~\cref{Lemma:reverse}.  
\end{proof}

\begin{figure}
\begin{center}
\includegraphics[scale=0.7]{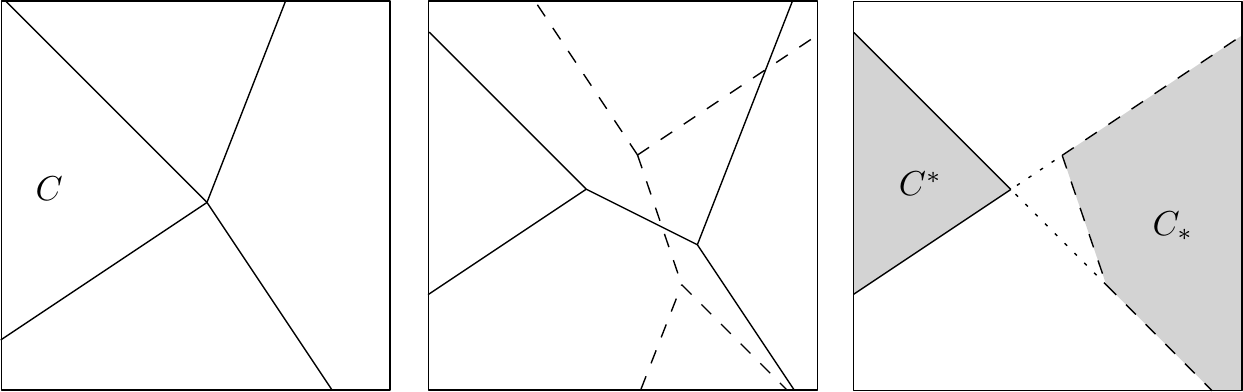}
\end{center}
\caption{A fan $\mathcal{F}$ (left). An instance of $\varphi^*(\mathcal{F},\omega)$ and $\varphi_*(\mathcal{F},\omega)$ (center). A cell $C_*$ contained in the reverse cone of $C^*$.}
\label{fig:powerdual}
\end{figure}


Let $\mathcal{F} \subset \R^d$ be a full-dimensional polyhedral fan with $n$ cells. 
We say that $\mathcal{F}$ is \emph{universally covering} if for any point set $P \subset \R^d$ of $n$ points there exists a covering assignment for $\mathcal{F}$ and $P$. 
After showing that all recursively regular fans are universally covering, one could imagine that all fans are so. 
We prove that this is not the case in dimension three and higher by showing that if a fan is cyclic in the sense described in the introduction, there is a point set for which there is no covering assignment. 
This statement will easily follow from \cref{line}. 
Before proving this theorem, we need to introduce a definition and state a technical lemma.

Let $\mathcal{F} \subset \R^d$ be a full-dimensional polyhedral fan, let $W \in \mathcal{F}$ be a facet incident with $ C , D \in \cells(\mathcal{F})$, and let $v$ be a vector normal to $W$ pointing from $C$ to $D$. 
We say that an assignment $ \sigma: \cells(\mathcal{F}) \to \R^d$ satisfies the \emph{overlapping condition} for $W$  if $ \scalprod{(\sigma(C) - \sigma(D))}{ v} \geq 0 $.
Note that the previous condition is satisfied for a facet and an assignment if and only if the copies of the two cells sharing the facet translated to the assigned points have non-empty intersection. 
We state now the following well-known facts, which are proved in the appendix.

\begin{restatable}{lemma}{lemdir}\label{lem:direction}
Let $K \subset \R^d$ be a full-dimensional polyhedral cone.
\begin{enumerate}[(i)]
\item Any line with direction interior to $K$ has unbounded intersection with $K$. \label{interior} 
\item Any line with direction not contained in $K$ has bounded intersection with $K$.\label{notin}
\end{enumerate}
\end{restatable}

The next lemma follows easily.

\begin{lemma}\label{lem:overlap}
Let $\mathcal{F} \subset \R^d$ be a full-dimensional polyhedral fan.
A covering assignment for $\mathcal{F}$ must satisfy the overlapping condition for every interior facet of the fan.
\end{lemma}

\begin{proof}
If the condition is not satisfied for the facet $H=C \cap D$, we consider a ray in a direction interior to $H$ (for instance, the barycenter of its rays) and placed at the point $(\sigma(C) + \sigma(D))/2$. 
In view of~\cref{lem:direction}-(\ref{notin}), no cell of $\mathcal{F}$, except for $C$ and $D$, can cover an unbounded part of this ray. 
In addition, none of these two cells intersect it. 
Therefore, since the ray is unbounded and we have finitely many cones, the ray cannot be completely covered. 
If the fan is complete, the proof is finished. 
Otherwise, we should note that the ray will eventually enter $|\mathcal{F}|$, since the direction of the ray is interior to an interior facet of $\mathcal{F}$ and, hence, interior to $|\mathcal{F}|$.
\end{proof}

The previous condition is not sufficient in general, not even in the plane. 
An exception is the case where all the points lie on a line, which is studied in the following lemma. 

\begin{lemma}\label{lem:overlapsuf}
Let $\sigma: \cells(\mathcal{F}) \rightarrow P$ be an assignment for a full-dimensional polyhedral fan $ \mathcal{F} \subset \R^d$ and a point set $P \subset \ell \cap |  \mathcal{F} | ^-$, where $\ell$ is a line. 
If $\sigma$ satisfies the overlapping condition, then it is a covering assignment. 
\end{lemma}

\begin{proof}
We prove first the complete case. 
Fix an orientation for $\ell$ and let $v$ be a direction vector for it. 
In addition, we can assume without loss of generality that $\ell$ goes through the apex of $\mathcal{F}$. 
Consider any oriented line $\ell'$ with direction $v$. 
At infinity, $\ell'$ is covered by some (untranslated) cell $C$ of $\mathcal{F}$. 
Hence, when $C$ is translated to its assigned point of $P$, it still covers $\ell'$ at infinity because the translation is only in the direction $v$. 
Let $q \in \ell'$ be the point where $\ell'$ leaves $C$.  
If $q$ is in the relative interior of (the translation of) a facet $W=C \cap D$, where $C,D \in \cells(\mathcal{F})$, the overlapping condition for $W$ (and the special position of $P$) ensures that $\ell'$ enters $D$ before leaving $C$. 
Iterating this argument, we eventually reach a cell containing the direction $-v$ that covers the unbounded remainder of~$\ell'$. 
Thus, any line $\ell'$ with direction $v$ and that intersects only $d$- and $(d-1)$-dimensional faces of the translated cells is completely covered. 
The union $\mathcal{U}$ of the remaining lines with direction $v$ (that is, the lines intersecting some $(d-2)$-dimensional face of some translated cell) is a nowhere-dense set and thus is covered as well. 
Indeed, for every line $ \hat{\ell} \in \mathcal{U}$ we can find a line not in $\mathcal{U}$ with direction $v$ (and, hence, covered) arbitrarily close to $\hat{\ell}$. 
Since the cells are closed sets, the limit of a sequence of covered lines must be covered as well, and thus $\mathcal{U}$ is covered. 
Since any line with direction $v$ is covered, $\R^d$ is completely covered.

Assume now that $\mathcal{F}$ is a conic fan with $K=|  \mathcal{F} |$.
Consider a line $\ell'$ with direction $v$ that enters $K$ through a facet. 
Let $C \in \cells(\mathcal{F})$ be the cell containing this facet. 
Since $P \subset \ell \cap  K^- $, the line $\ell'$ should enter the cell $C$ translated to the corresponding point before entering $K$. 
The arguments for the complete case carry over until the line crosses (the translation of) a facet $W$ of a cell $D$ such that $W \subset \partial K$. 
Then, again the fact that $P \subset \ell \cap  K^- $ implies that the $\ell$ had left $K$ before. 
Therefore, if $\ell'$ is a line with direction $v$ that avoids $(d-2)$-dimensional faces of the translated cells (and of $K$), then $\ell' \cap K$ is covered. 
A limit argument as in the complete case ensures that then all the lines with direction $v$ has the portion intersecting $K$ covered, and thus $K$ is covered.
%
%
%
\end{proof}

It can be proved that if the overlapping conditions are satisfied for an assignment of a $2$-dimensional fan, then the uncovered region is a convex polygon.
In addition, it can be tested whether this polygon is empty (and, thus, if the assignment is covering) in time linear in the number of cells of the fan (if the adjacency information is in the input). 

We are now in a position to construct examples consisting of a fan and a point set for which there is no covering assignment.

\begin{theorem}
\label{line}
Given a full-dimensional polyhedral fan $ \mathcal{F} \subset \R^d$ with $n$ cells and set of $n$ points $P \subset \ell \cap |  \mathcal{F} | ^-$, where $ \ell $ is a line, there is a covering assignment for $\mathcal{F}$ and $P$ if and only if $\mathcal{F}$ is acyclic in the  direction of $\ell$. 
\end{theorem}

\begin{proof}
Provided that $ \mathcal{F} $ is acyclic in the direction $ v $ of $ \ell $, we can construct a directed acyclic graph having the cells of $ \mathcal{F} $ as vertices and an edge from $D$ to $C$ if the vector $u$ normal to $W =C \cap D $ pointing from $C$ to $D$ satisfies $\scalprod{u}{ v} \geq 0 $. 
If the order as the points $ \sigma(C) $ appear on $ \ell $ (for all $C \in \cells(\mathcal{F})$) respects the partial order represented by such a directed graph, then the overlapping condition holds for $ \sigma $. 
\cref{lem:overlapsuf} ensures that this condition is sufficient for the assignment to be covering.

We prove the other direction by contrapositive. 
If there is a visibility cycle $ \tau = (C_1 \dots C_k)$ in the direction $ v$ (that is, $C_{i}$ is in front of $C_{i+1}$, for all $i \in [k-1]$, and $C_k$ is in front of $C_1$), there is a cycle in the order the points $ \sigma(C_1), \dots, \sigma(C_k) $ should appear in the line, preventing the overlapping condition to be satisfied for all the facets of the fan. 
This has been proven to be necessary for the assignment to be covering. 
\end{proof}



If a covering assignment exists for a given point set in a line and a given fan, it can be computed in $ O(n^2) $ time by performing a topological sort on the graph described in the proof of \cref{line}.  
Since the number of facets is bounded by $n^2$, the algorithm runs in the claimed time. 
Afterwards, it only remains to sort the points, which can be done in $O(n \log n)$ time.  
Moreover, the topological sort algorithm would detect if the graph has a cycle and, therefore, there is no covering assignment. 

\begin{figure}
\begin{center}
\includegraphics[scale=0.7,page=3]{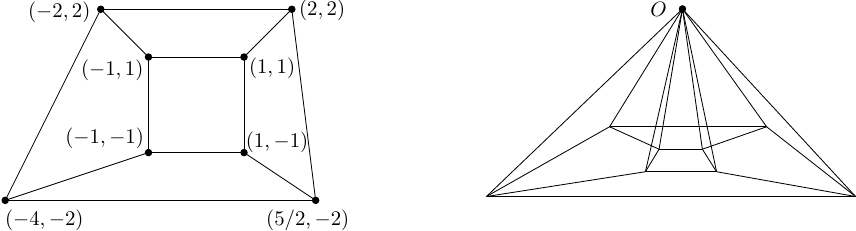}
\end{center}
\caption{Acyclic, not recursively-regular subdivision (left) and the corresponding fan (right).}
\label{fig:counter}
\end{figure}

After understanding the previous theorem, one might be tempted to conjecture that being acyclic is equivalent to being universally covering. 
We exhibit next an example to show that this is not the case.

\begin{proposition}
\label{prop:counter}
There exists full-dimensional polyhedral fan $\mathcal{F} \subset \R^3$ consisting of $n$ cells, and a set of $n$ points $P \subset |\mathcal{F}|^-$ for which there is no assignment satisfying the overlapping conditions. 
The fan has no cycle in any direction.
\end{proposition}
\begin{proof}
We will provide a three-dimensional fan $\mathcal{F}$ with five cells and a point set $P \subset \R^3$ for which there is no covering assignment. 
More precisely, it can be shown that for each of the $5!$ possible assignments, one of the eight overlapping conditions is violated. To construct $\mathcal{F}$, take the subdivision sketched in \cref{fig:counter} (left) and embed it in the plane $\{(x,y,z) \in \R^3:z=-1/8 \}$. Take then the cones from the origin to each of the cells of this subdivision forming the fan displayed in~\cref{fig:counter} (right). 

Let $P$ be the point set consisting of the points
\begin{align*} 
 p_1=(29,95,89), \; p_2=&(55,19,92), \; p_3=(54,10,82)\\
 p_4=(78,2,68),& \; p_5=(15,40,92).
\end{align*}
There is no assignment for this point set fulfilling all the overlapping conditions, as proved in~\cref{App:calculations}. 
The last statement together with \cref{lem:overlapsuf} allow us to derive that there is no covering assignment for the given fan and the given point set. 
That there is no direction in which $\mathcal{F}$ is cyclic is also proven in~\cref{App:calculations}. 

\end{proof}

The point set $P$ in the previous proof was found with the help of a computer. 
We generated many pseudo-random samples of five points in $\R^3$ trying different precisions for the coordinate generator and several parameters for the distribution.

This last example motivates the conjecture that a fan is covering if and only if it is recursively regular. 
Note that a fan that is not recursively-regular must have a completely non-regular convex region, and this fact could perhaps be used to construct a point set for which no covering assignment exists. 

\paragraph{Illuminating a stage.}
The problem of illuminating a pointed cone using floodlights is closely related to the problem of illuminating a stage considered in~\cite{BGAMSU,Czyzowicz,Dietel,Hiro}. 
Informally, the problem in the plane asks whether given $n$ angles and $n$ points, floodlights having the required angles can be placed on the points in a way that a given segment (the stage) is completely illuminated. 
The problem can be generalized to higher dimensions where our results on covering a cone by a conic fan have new implications (see~\cite{mythesis}).

\section[Other applications and related problems]{Other applications and related problems}
\label{aps}

In this section we describe applications of the theoretical results introduced before. 

\subsection{Redundancy in spider webs}  
We present now a problem in tensegrity theory related to the finest regular coarsening of subdivisions in~$\R^2$. 
We first review the main results we will need.

\paragraph{The Maxwell-Cremona correspondence.}
\label{MaxCrem}

Tensegrity theory studies the rigidity properties of frameworks made of bars, cables and struts from a formal point of view. 
An \emph{abstract framework} $G=(V;B,C,S)$ is a graph on the vertex set $V=\{v_1,\dots,v_n\}$ whose edge set $E$ is partitioned into sets $B$, $C$ and $S$. 
The edges in $B$ are called \emph{bars}, the ones in $C$ are called \emph{cables} and the ones in $S$ are called \emph{struts}. 
They represent links supporting any stress, non-negative stresses and non-positive stresses, respectively. 
A \emph{(tensegrity) framework} (in $\R^2$) is an abstract framework together with an embedding of the vertices $p:V \rightarrow \R^2$ where we put $p(v_i)=p_i$, for $i \in [n]$. 
The framework will be denoted by $G(p)$ and $p$ will be thought of as a point $(p_1,\dots,p_n) \in \R^{2n}$. 
We can consider the configuration space of $G(p)$ to be 
\begin{align}
X(p)=\{(x_1,\dots,x_n) \in \R^{2n} : \| x_i-x_j\| & =\| p_i-p_j \|  \text{, for all } v_iv_j \in B;  \nonumber  \\
 \| x_i-x_j\| & \leq \| p_i-p_j \|  \text{, for all } v_iv_j \in C;  \nonumber \\	
 \| x_i-x_j\| & \geq \| p_i-p_j \|  \text{, for all } v_iv_j \in S\}.\label{rigid}
\end{align}
That is, $X(p)$ is the set of embeddings of $G$ preserving the length of the bars, making the lengths of the cables no longer and the lengths of the struts no shorter than their lengths induced by $p$. 

A tensegrity framework $G(p)$ is \emph{rigid} in $\R^d$ if there exists an open neighborhood $ U \subset \R^{2n}$ of $p$ such that $X(p) \cap U=M(p) \cap U$, where $$M(p) = \{ (x_1,\dots,x_n) \in \R^{2n} : \| x_i-x_j\| =\| p_i-p_j \| \text{, for all }  i,j \in [n] \}$$ is the manifold of rigid motions associated to $p$. 
%
In other words, a framework is rigid if its only motions respecting the constraints (\ref{rigid}) are the motions that rigidly move the whole framework. 
The study of the quadratic constraints in the definition of $X(p)$ can be complicated. 
Because of this, the notion of infinitesimal rigidity was introduced, which captures the rigidity constraints up to the first order. 
Consider the system of linear equations and inequalities obtained by differentiating the constraints in (\ref{rigid}). 
If the solutions of the system correspond only to differentials of motions in the Euclidean group, the framework is \emph{infinitesimally rigid}. 
It is known that infinitesimal rigidity implies rigidity and that the converse is in general not true. 

Given a framework $G(p)$, we say that $\omega: E \to \R$ is a \emph{proper (equilibrium) stress} for $G(p)$ if the following conditions hold:
\begin{enumerate}[(1)]
\item $\omega(v_iv_j)=0$ if $v_iv_j \not \in E$.
\item $\omega(v_iv_j) \geq 0$ if $v_iv_j \in C$.
\item $\omega(v_iv_j) \leq 0$ if $v_iv_j \in S$.
\item Every $v_i \in V$ is \emph{in equilibrium}. That is, $ \underset{v_j \in V} \sum \omega(v_iv_j) ( p_j-p_i )=0 $.
\end{enumerate}
We say that $\omega$ is \emph{strictly proper} if the stresses on all cables and struts are non-zero.
%
Intuitively, $\omega$ is a proper equilibrium stress for $G(p)$ if the forces exerted by the edges (represented by $\omega$) on the vertices add up to zero, taking into account that cables can support only non-negative stresses and struts can support only non-positive ones. 
Clearly, the stress assigning zero to all the edges is proper. 
This stress is called the \emph{trivial stress}.


We state now a the Maxwell-Cremona correspondence, referring to~\cite{CrapoWhiteley} for more details.

\begin{theorem}[Maxwell-Cremona correspondence]
\label{MCcorresp}
Let $G$ be an abstract framework and $G(p)$ be a planar straight-line realization of $G$. 
There is a bijection between proper stresses for $G(p)$ and polyhedral terrains (with one arbitrarily chosen but fixed face at height zero) projecting on~$G(p)$, where positive stress values correspond to valleys, negative stress values correspond to mountains and zero stress values correspond to flat edges in the lifting. 
\end{theorem}

A \emph{spider web} is a framework (in $\R^2$) whose graph is connected, consisting only of cables, and with the vertices in the convex hull pinned down (that is, in equilibrium by definition). 
The two following results relate equilibrium stresses of a framework with its rigidity and infinitesimal rigidity. 

\begin{lemma}[Connelly \cite{Con82}]
\label{Conn}
If a spider web has a strictly proper stress, then it is rigid.
\end{lemma}

\begin{lemma}[Roth and Whiteley \cite{RW}]
\label{RothWhi}
If a tensegrity framework is infinitesimally rigid, then it has a strictly proper stress.
\end{lemma} 

Let $G$ be an abstract spider web on the vertex set $V$, and let $p:V \to \R^2$ be an embedding corresponding to a non-crossing straight-line realization of $G$. 
Assume that the vertices lying on the convex hull of $p(V)$ are fixed (they are, therefore, in equilibrium by definition). 
Note that the straight-line realization of $G(p)$ can be thought of as a polyhedral subdivision of the convex hull of $p(V)$ in the plane. 
Throughout this section, this subdivision will be denoted by~$\mathcal{S}=\mathcal{S}(G(p))$ and called the subdivision \emph{associated to} the spider web.

The Maxwell-Cremona correspondence states that $G(p)$ has a strictly-proper stress if and only if $\mathcal{S}$ is regular. 
From this fact, it is easy to derive the following proposition.

\begin{proposition}
Let $\mathcal{S}$ be the subdivision associated to a planar spider web $G(p)$. 
\begin{enumerate}[(i)]
\item Only the cables of $G$ corresponding to edges of the finest regular coarsening of $\mathcal{S}$ support a positive stress in any equilibrium stress of $G(p)$.
\item If $\mathcal{S}$ is recursively regular, then $G(p)$ is rigid.
\end{enumerate}
\end{proposition}

\begin{proof}
\hfill
\begin{enumerate}[(i)]
\item Since we showed that the edges omitted in the finest regular coarsening are lifted into a plane by any convex lifting (\cref{thm:coarsening}), the Maxwell-Cremona correspondence indicates that the corresponding cables will receive no stress in any proper equilibrium.
\item The finest regular coarsening of the subdivision corresponds to a set of cables such that there is an equilibrium stress assigning positive values to all of them. 
Therefore, the spider web defined by this set of cables is rigid by \cref{Conn}. 
For each of the subsubdivisions defined by the finest regular coarsening, we can assume that the vertices in the corresponding convex hull are now fixed and apply the previous argument recursively. 
\qedhere
\end{enumerate}
\end{proof}

\cref{Fig:spider} illustrates the previous result. 
The spider web represented in it is constructed from a triangulation appearing in~\cite{TriangInterNice}.
The edges omitted in the picture to the right, which do not belong to the finest regular coarsening of the  associated subdivision, support no stress in any equilibrium. 
Therefore, they can be considered redundant.

\begin{figure}
\begin{center}
 \includegraphics[scale=0.7]{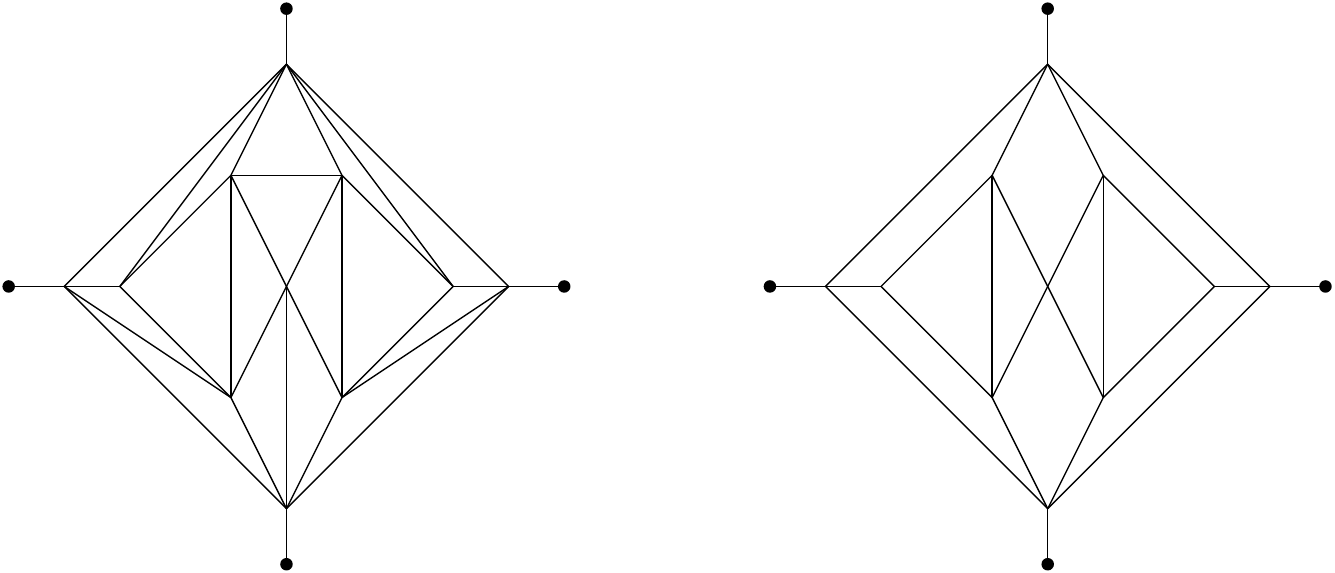}
\end{center}
\caption{A spider web (left) and the result of removing redundant cables (right).}
 \label{Fig:spider}
\end{figure}

Note that even thought recursively-regular subdivisions are associated to rigid spider webs, these might be far from infinitesimally rigid.
For instance, if a regular subdivision is refined by adding an edge whose endpoints are interior to previous edges, the result is recursively regular but obviously not infinitesimally rigid. 
We next translate a well-known fact of infinitesimal rigidity to the language of finest regular coarsenings.

\begin{corollary}
The subdivision associated to a infinitesimally rigid spider web is its own finest regular coarsening (hence, it is regular).
\end{corollary}

\begin{proof}
As \cref{RothWhi} states, if a framework is infinitesimally rigid, it has a strictly-proper stress. 
The edges omitted in the finest regular coarsening of the associated subdivision cannot participate in such stress. 
Therefore, none of the edges are omitted in the finest regular coarsening of the subdivision. 
\end{proof}

\subsection{Embeddings of directional graphs}

As shown in \cref{floodlights}, for the existence of a covering assignment it is necessary that there is an assignment satisfying the overlapping condition for every interior facet of the fan. 
Moreover, the examples we have found so far of polyhedral fans and point sets for which there is no covering assignment fail to fulfill the second condition. 
Hence, it could be that this condition is also sufficient. 
In any case, we think that it is of independent interest to study this condition alone, which is connected to a problem on graph embedding. 

Note first that the overlapping condition for a facet can be expressed as a requirement on the order in which the two involved points are swept by a hyperplane parallel to the facet. 
That is, we want to know which of two points ``appears first'' in a specific direction. 
The problem we study here asks whether, given set of relations of this type (stated on labels) and a point set, we can find a one-to-one labeling of the point set such that every relation is satisfied. 
We next describe the problem formally.

A \emph{directional graph} is a tuple $\overrightarrow{G}=(V,h)$, where $V$ is a set and ${ h:V \times V \rightarrow \R^d}$ is a function such that $h(v,u)=-h(u,v)$, for all $v,u \in V $. 
The elements of $V$ are called \emph{vertices}. 
We say that $u,v \in V$ are connected by an \emph{edge} if $h(v,u) \neq 0$. 
The \emph{dimension} of $\overrightarrow{G}$ is $d$. 
%
We may regard this structure as a directed graph with a non-zero direction associated to every edge. 
Such a graph will be called the \emph{underlying graph} of the directional graph. 
Note that the condition in the definition already implies that $h(v,v)=0$, for all $v \in V$. 

\begin{figure}
\begin{center}
\includegraphics[page=2,scale=0.7]{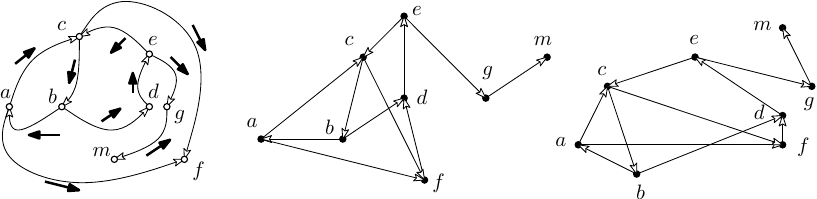}
\end{center}
\caption{A directional graph (left), a drawing (center), and an embedding (right).}
\label{fig:embedd}
\end{figure}

An \emph{embedding} of a $d$-dimensional directional graph  $\overrightarrow{G}=(V,h)$ on a point set $P \subset \R^d$ is a one-to-one assignment $\sigma:V \rightarrow P$ such that $$ \scalprod{h(v,u)}{ \sigma(v)-\sigma(u)} \geq 0 \text{, for all } v,u \in V.$$ 
If such an embedding exists, we say that $\overrightarrow{G}$ is \emph{embeddable} in $P$. 
A \emph{drawing} of a directional graph $\overrightarrow{G}=(V,h)$ is a bijection $\pi:V \rightarrow S \subset \R^d$ such that for all $u,v \in V$ with $h(u,v) \neq 0$ we have that $\pi(v)-\pi(u) = \lambda _{uv} \cdot h(v,u)$ for some $\lambda _{uv}>0$.
The \emph{projection} of a $d$-dimensional directional graph $\overrightarrow{G}$ into a $k$-dimensional linear subspace $ L \subset \R^d$ is the $k$-dimensional directional graph obtained by projecting the vector $h(u,v) \in \R^d$ onto $L \cong \R^k $, for all $v,u \in V$.

A directional graph is illustrated in \cref{fig:embedd}, together with a drawing and an embedding. 
The arrows near the edges indicate the directions associated with them. 
Observe that the embedding condition for an edge restricts its direction to a halfspace, while the drawing condition fixes its direction completely.  
Note also that the lengths of the vectors assigned by $h$ are irrelevant for the existence of an embedding or a drawing of a directional graph. 
Therefore, we will consider two directional graphs $(V,h)$ and $(V,h')$ \emph{equivalent} if $h(u,v)$ is a positive scalar multiple of $h'(u,v)$ for all $u,v \in V$.

A $d$-dimensional directional graph $\overrightarrow{G}=(V,h)$ is \emph{universally embeddable} if it is embeddable on any point set $P \subset \R^d$ with $|P|= |V|$ . 
It is \emph{drawable} if it has a drawing.
%
%

The \emph{directional graph of a polytope} is the set of its vertices, together with the function $h(u,v)=v-u$ if $u$ and $v$ are endpoints of an edge of the polytope, and $h(u,v)=0$ otherwise. 
The \emph{normal graph} of a polyhedral fan is set of its cells with the function $h(C,D)$ being a vector normal to the facet common to $C$ and $D$ and pointing ``from $C$ to $D$ '' if they share a facet, and $h(u,v)=0$ otherwise.
%
Note that the directional graph of a polytope and the graph of its normal fan are embedding-equivalent. 
This is a consequence of the duality between a polytope at its normal fan. 
%
%
The following proposition shows that there is a surprisingly large family of universally embeddable directional graphs.

\begin{proposition}
\label{prop:draw}
If a directional graph is drawable, then it is universally embeddable. 
In particular, a directional graph $\overrightarrow{G}=(V,h)$ with underlying graph being a tree is universally embeddable regardless of $h$. 
The directional graph of a polytope is universally embeddable. 
\end{proposition}

\begin{proof}
Given a drawable directional graph $\overrightarrow{G}=(V,h)$ and an arbitrary point set $P $ with $|  P | = |  V | $, consider a drawing $\pi$ of $\overrightarrow{G}$. 
Let $\mu$ be the least-squares optimal matching between $\pi(V)$ and $P$. 
We will show that $\mu \circ \pi$ is an embedding of $\overrightarrow{G}$. 
Assume that it is not the case. 
Then, there must be a pair $u,v \in V$ such that $ \scalprod{h(v,u)}{\mu(\pi(v))-\mu(\pi(u))} < 0$.
Since $\pi(u)-\pi(v)=\lambda_{uv} \cdot h(v,u)$, for some $\lambda_{uv} \in \R^+$, we have that $\scalprod{\pi(u)-\pi(v)}{\mu(\pi(v))-\mu(\pi(u))} <0$, which contradicts the optimality of $\mu$ because swapping the images of $\pi(u)$ and $\pi(v)$ would improve the matching. 
Directional graphs having a tree as underlying graph are trivially drawable and directional graphs of polytopes have the $1$-skeleton of the polytope as a drawing. 
\end{proof}

It is not hard to see that if there is a sequence of vertices $v_1,\dots,v_l,v_{l+1}=v_1$ in $V$ and a vector $\delta \in \R^d$ such that $\scalprod{h(v_i,v_{i+1})}{ \delta }>0$, for all $i \in [l]$, then the graph is not drawable. Such a cycle is called a \emph{($\delta$-)forcing cycle}. However, the converse is not true in general: for instance, the normal graph of the subdivision in \cref{fig:counter} has no forcing cycle but it is also non-drawable.  


The following proposition summarizes some relations of recursive regularity to drawability and embeddability of directional graphs.

\begin{proposition}~\hfill
\begin{enumerate}[(i)]
\item The projection of a universally embeddable directional graph is universally embeddable.
\item Normal graphs of recursively-regular fans are universally embeddable.\label{cond3}
\item Universally embeddable graphs are not necessarily drawable.
\item Graphs with forcing cycles are not universally embeddable. 
\item There are graphs with no forcing cycles that are not universally-embeddable.
\end{enumerate}
\end{proposition}

\begin{proof}
\hfill
\begin{enumerate}[(i)]
\item Let $\overrightarrow{G}=(V,h)$ be a $d$-dimensional universally-embeddable directional graph, and let $L$ be a $k$-dimensional linear subspace of $\R^d$ with a basis $\{l_1,\ldots,l_k\}$. 
Let $\bar{G}=(V,\bar{h})$ be the projection of $\overrightarrow{G}$ onto $L$, which is identified with $\R^k$ through the bijection  
\begin{align*}
i: \R^k & \to L \subset \R^d \\
(x_1,\ldots, x_k) & \longmapsto \sum_{j \in[k] } x_j l_j. 
\end{align*}
Consider any set of $|V|$ points $\bar{P} \subset \R^k$, and the associated point set $P=i(\bar{P}) \subset \R^d$. 
If $\sigma: V \to P$ is an embedding of $\overrightarrow{G}$ on $P$, then $\bar{\sigma}= i^{-1} \circ \sigma $ is an embedding of $\bar{G}$ in $\bar{P}$, where~$i^{-1}$ denotes the inverse of $i$ on $L$. 
Indeed, $\scalprod{h(v,u)}{ \sigma(v)-\sigma(u)}=\scalprod{\bar{h}(v,u)}{ \bar{\sigma}(v)-\bar{\sigma}(u)}$ for all $u,v \in V$, because $\sigma(u)-\sigma(v) \in L$ and thus only the projection of $h(u,v)$ onto $L$ contributes to the scalar product. 
\item Let $\mathcal{F} \subset \R^d$ be a full-dimensional polyhedral fan consisting of $n$ cells.
\cref{thm:recreg} ensures that there is a covering assignment for $\mathcal{F}$ and any set $P$ of $n$ points. 
This assignment must satisfy the overlapping condition for each facet of the fan, which is equivalent to the embedding condition for the corresponding edge. 
\label{normue}
\item The normal graph of a fan is drawable if and only if the fan is regular (see, for instance, \cite{AurenDual}). 
Thus, the normal graph of a recursively-regular non-regular fan is not drawable and it is, however, universally embeddable, as shown in~(\ref{normue}). 
\item Consider a $\delta$-forcing cycle $v_1,\dots,v_l, v_{l+1}=v_1$. Take a set of different points in a line having direction vector $\delta$ and label them increasingly with respect to their scalar products with $\delta$. For any embedding $\sigma$, $\sigma(v_{i+1})$ must have a label larger than $\sigma(v_{i})$, for all $i \in [l]$, which is obviously impossible.
\item The normal graph of the fan obtained by taking cones from the subdivision in \cref{Fig:anrr} has no forcing cycle, since it is acyclic (in the visibility sense). 
However, we have given a set of points for which all the assignments violate an overlapping condition. Hence, there is no embedding of its normal graph into this point set. \qedhere
\end{enumerate}
\end{proof}

\section[Concluding remarks and open problems]{Concluding remarks and open problems}

We have shown that the finest regular coarsening of a subdivision, which can be seen as the regular subdivision that is closest to it, can be used to define a structure called the regularity tree. 
The leaves of this tree define a partition of the subdivision in sub-subdivisions that are either regular or completely non-regular. 
The regularity tree reflects thus some of the structure of non-regular subdivisions and measures, in a sense, the degree of regularity. 
As a consequence, the class of recursively-regular subdivisions arises in a natural way. 
We have shown that this class goes beyond regular subdivisions while excluding cyclic ones. 
However, we have proven that they are in general not connected by flips.

In addition, we have studied a collection of related applications, and
we expect to find even more, since any theorem or algorithm based on the regularity of a subdivision and admitting a recursive scheme can probably be extended to apply for the larger set of recursively-regular subdivisions.

In particular, we have focused on the problem of illuminating the space by floodlights. 
It was known that regular fans are universal and our aim was to answer the question for the other fans. 
We have proved that not only regular fans are universal and that not only cyclic ones are non-universal. 
It makes then sense to ask what is the complexity class of the general problem of deciding whether the space can be covered by a given fan from a given point set (in dimensions bigger than two). 
It remains open as well to precise the limits of universality, that is, to characterize the polyhedral fans that can cover the space from any point set. 
A reasonable candidate is recursive-regularity. 
Indeed, the fact that a non-recursively-regular subdivision has a convex sub-subdivision which is completely non-regular could be the first step towards a proof for this fact. 
Our results on covering the space by floodlights have implications for a three-dimensional version of the \emph{stage illumination problem}.
In data visualization, recursive partitions using regular subdivisions (Voronoi Treemaps~\cite{BD}) have been used to visualize hierarchical structures. 
Although these partitions are not polyhedral subdivisions, they can be constructed from a recursively-regular subdivision applying a weighting scheme as in the proof of~\cref{thm:recreg}~\cite{mythesis}. 

The problem of embedding directional graphs is left in a similar situation. 
A natural and easy to state open question is whether deciding if a directional graph can be embedded in a given point set is NP-hard. 

Concerning algorithmic issues, we have proven that the finest regular coarsening and the regularity tree of a subdivision can be computed in polynomial time. 
We have used these facts to prove that recursive regularity of a subdivision can be decided in polynomial time as well, which is relevant for the algorithmic version of the aforementioned problems. 


{\small
\bibliographystyle{abbrv}
\bibliography{BT3}
}

\appendix

\section{Proof of~\cref{lem:GordanExt}}\label{proofGordan}


\gordan*

\begin{proof}
If $E=\emptyset$, the statement of the theorem is Gordan's theorem. 
Hence, assume without loss of generality that $E=[j]$ for some $j \in [m]$. 
We will reduce this case to Gordan's theorem. 
Assume further that the span $L$ of the first~$j$ row vectors is $k$-dimensional and that the first $k$ rows span $L$. 
Let $M'$ be the matrix resulting of excluding from $M$ the rows indexed by $[j] \setminus [k]$. 
The set of solutions of $S(M',[k])$ is exactly the same as the set of solutions of  $S(M,[j])$. 
On the other hand, we show next that the system $S^*(M,[j])$ has a solution if and only if the system $S^*(M',[k])$ has. 
Although the dimension of the domain of the linear map associated to $(M')^\top$ is smaller than the one of $M^\top$, a solution to $S^*(M',[k])$ can be extended to a solution of $S^*(M,[j])$ by setting the remaining coordinates to zero. 
That is, the dimension of the kernels of the linear maps associated to both matrices are the same. 

Let now $y_0$ be a solution to $S^*(M,[j])$. 
We can obtain another solution having the coordinates indexed by $[j] \setminus [k]$ equal to zero, by expressing the columns of $M$ indexed by $[j] \setminus [k]$ as linear combinations of the columns indexed by $[k]$, and modifying the coefficients in~$[k]$ accordingly. 
In this way, ignoring the coordinates indexed by $[j] \setminus [k]$, we obtain a solution to $S^*(M',[k])$. 

Henceforth, we will assume then that $E=[k]$ and that the first $k$ rows of $M$ are linearly independent. 
The set of equations in $S(M,[k])$ restricts then the variables to an $(n-k)$-dimensional linear subspace of~$\R^n$, for $0 \leq k < n $. 
We can then find an invertible $n \times n$ matrix $T$ such that the space defined by the equations of $S(MT,[k])$ is the one having the first $k$ coordinates equal to zero. 
Since $T$ is invertible, $S(M,[k])$ has a solution if and only if the system $S(MT,[k])$ has.

On the other hand, $S^*(M,[k])$ and $S^*(MT,[k])$ have the same set of solutions because $T$ is invertible and, thus, $M^\top y=0$ if and only if $T^\top M^\top y=0$ for all $y \in \R^m$. 


Assume now that 
\[
MT=
\left(
\begin{array}{cc}
\text{Id}_{k \times k} & 0_{k \times (n-k)}\\ 
R & N
\end{array}\right), 
\]
with $R \in \R^{k' \times k}$, $N \in \R^{k'\times (n-k)}$, and  $k'=m- k $.
The systems $S(MT,[k])$ and
\[
S(N):
\begin{cases}
\begin{array}{l}
N y >0 \\
y \in \R^{n-k}
\end{array}
\end{cases}
\]
have both a solution or none of them have. 
This is because a solution to the first must have the first $k$ coordinates equal to zero and, thus, the last $n-k$ coordinates must be a solution of the second. 
Conversely, a solution to the second system can be extended to a solution of the first by just adding $k$ zero coordinates. 

Similarly, the system $S^*(MT,[k])$ and the system
\[
S^*(N):
\begin{cases}
\begin{array}{l}
N^\top y =0 \\ 
y \in \R^{k'} \\ 
y \geq 0, \; y \neq 0
\end{array}
\end{cases}
\]
have both a solution or none of them has. 
Indeed, the restriction of a solution of $S^*(MT,[k])$ to the last $k'$ coordinates must be a solution of $S^*(N)$.
For the other direction, let  
\[ z_0 = \left( 
\begin{array}{c}
-R^\top y_0 \\
y_0
\end{array} \right),
\]
where $y_0$ is a solution of $S^*(N)$. 
Since
\[
(MT)^\top z_0 = 
\left(
\begin{array}{cc}
\text{Id}_{k \times k} & R^\top\\ 
0_{(n-k) \times k} & N^\top
\end{array}
\right) \left( 
\begin{array}{c}
-R^\top y_0 \\
y_0
\end{array} \right) = 	\left( 
\begin{array}{c}
-R^\top y_0+ R^\top y_0 \\
0+N^\top y_0
\end{array} \right)	=0, 
\]
the vector $z_0$ is a solution of $S^*(MT,[k])$.
%
%
We finish the proof by applying Gordan's theorem to $S(N)$ and $S^*(N)$.
\end{proof}

\section{Proof of~\cref{lem:direction}}\label{prooflemdirection}


\lemdir*

\begin{proof} \hfill 
\begin{enumerate}[(i)]
\item Let $\ell= \{ p + \lambda v: \lambda \in \R \}$ be a line passing through $p\in\R^d$ and with direction $v \in \R^d$ interior to $K=\{ q + \sum_{i \in I} \alpha_i v_i: \alpha_i \ge 0 \text{ for all } i \in I  \}$, where $v_i$ for $i\in I$ are the extreme rays of $K$. 
Since $v$ is interior to $K$, we can express $v= \sum_{i \in I} \gamma_i v_i$ with $\gamma_i >0$ for all $i \in I$. 
The set of vectors $ \{v_i: i \in I \}$ spans $\R^d$ because its positive span $K$ is not contained in any proper affine subspace. 
Thus, we can express $p-q= \sum_{i \in I} \delta_i v_i $, where $\delta_i \in \R$ for all $i \in I$. 
Note now that 
\[ p +\lambda v = q + (p-q) + \lambda \sum_{i \in I} \gamma_i v_i = q+ \sum_{i \in I} ( \delta_i + \lambda \gamma_i ) v_i,\]
which lies in $K$ for all $\lambda \ge  \max_{i \in I} \frac{ -\delta_i}{\gamma_i}$.
\item We triangulate $K$ into a finite number of simplicial cones and show that the intersection of $\ell$ with each cone is bounded. 
For a fixed simplicial cone $K'$ with extreme rays $v_1,\ldots,v_d$, the direction $v$ of the line can be expressed in a unique way as $v=\sum_{i \in [d]} \gamma_i v_i$ with $\gamma_i \in \R$ for all $i \in [d]$. 

We will prove the contrapositive. 
Assume that there exists $\lambda_0 \in \R$ such that $p+\lambda v \in K'$, for all $\lambda \ge \lambda_0$. 
Since 
\[ p+\lambda v=p+\lambda \sum_{i \in [d]} \gamma_i v_i= q+ \sum_{i \in [d]} (\lambda \gamma_i-\delta_i) v_i,\]
we have that $\lambda \gamma_i-\delta_i \ge 0$ for all $i \in [d]$ and for all $\lambda \ge \lambda_0$. 
Thus, $\gamma_i \ge 0$ for all $i \in [d]$,
which implies that the direction of $v$ is contained in $K'$. \qedhere
\end{enumerate}
\end{proof}

\section{A non-universal acyclic polyhedral fan}\label{App:calculations}

We present here the calculations for the proof of \cref{prop:counter}. 
Recall that the point set from the counterexample is 
\begin{align*} 
 p_1=&(29,95,89) \\
 p_2=&(55,19,92)\\
 p_3=&(54,10,82)\\
 p_4=&(78,2,68)\\
 p_5=&(15,40,92).
\end{align*}

The fan involved is drawn in \cref{fig:counterLabeled}, where the facets and cells of the fan have been labeled and it has been truncated. From the coordinates of the points in this figure and assuming that we embed it into the plane $\{z=-1/8 \}$ and we cone the cells with the origin, normal vectors to the facets of the fan can be computed:
\begin{align*} 
v_{12}&= (   4         ,  0 ,       -32)    \\
v_{13}&= (2,2,0)\\
v_{15}&=  (   1     ,   -3    ,    16)\\
v_{23}&=   (     0    ,    4   ,     32)\\
v_{24}&=    (    4   ,        0  ,      32)\\
v_{25}&=     (     0    ,    -4   ,     32)\\
v_{34}&=     (  2    ,     -2    ,       0)\\
v_{45}&=      (  -2     ,    -3     ,   8).
\end{align*}

The first column of the table visits the $5!$ permutations representing all possible assignments from cells to points. The notation used for the permutations is simply the concatenation of the labels of the points assigned to $C_1,C_2,C_3,C_4$ and $C_5$ in this order. 
The second row indicates one facet for which the corresponding assignment does not satisfy the overlapping condition. 
The two following columns just extract the points involved in the violation. 
The last column computes the ``gap'' between the two translated floodlights. The positivity of this last value certifies that the overlapping condition is not fulfilled. 

\begin{figure}[ht]
\begin{center}
\includegraphics[scale=0.7,page=2]{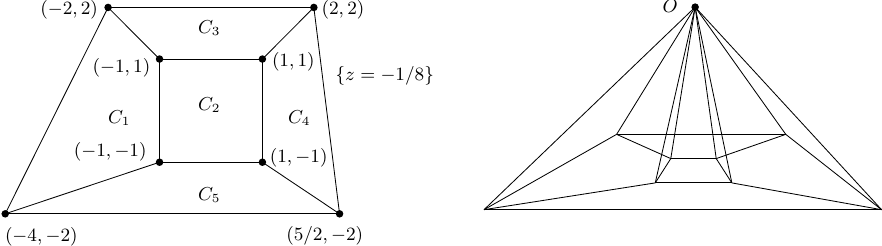}
\end{center}
\caption{The fan of the counterexample and a section of it.}
\label{fig:counterLabeled}
\end{figure}

\newpage
\begin{longtable}{c c c c l}
$\sigma$ & $f_{ij}$ & $\sigma(C_i)$ & $\sigma(C_j)$ & $(\sigma(C_i)-\sigma(C_j)) \cdot v_{ij} $\\
\hline

\endhead

\hline & & & & Continued on next page \\
\endfoot

\hline \hline
\endlastfoot

$5  4  3  2  1 $ & $f_{12}$  & $p_{5}$ & $p_{4}$ & $  ( p_{4} - p_{5} )   \cdot  v_{12}  = (63,-38,-24) \cdot (4,0,-32) = 1020$ \\ 
 $5  4  3  1  2 $ & $f_{12}$  & $p_{5}$ & $p_{4}$ & $  ( p_{4} - p_{5} )   \cdot  v_{12}  = (63,-38,-24) \cdot (4,0,-32) = 1020$ \\ 
 $5  4  2  3  1 $ & $f_{12}$  & $p_{5}$ & $p_{4}$ & $  ( p_{4} - p_{5} )   \cdot  v_{12}  = (63,-38,-24) \cdot (4,0,-32) = 1020$ \\ 
 $5  4  2  1  3 $ & $f_{12}$  & $p_{5}$ & $p_{4}$ & $  ( p_{4} - p_{5} )   \cdot  v_{12}  = (63,-38,-24) \cdot (4,0,-32) = 1020$ \\ 
 $5  4  1  2  3 $ & $f_{12}$  & $p_{5}$ & $p_{4}$ & $  ( p_{4} - p_{5} )   \cdot  v_{12}  = (63,-38,-24) \cdot (4,0,-32) = 1020$ \\ 
 $5  4  1  3  2 $ & $f_{12}$  & $p_{5}$ & $p_{4}$ & $  ( p_{4} - p_{5} )   \cdot  v_{12}  = (63,-38,-24) \cdot (4,0,-32) = 1020$ \\ 
 $5  3  4  2  1 $ & $f_{12}$  & $p_{5}$ & $p_{3}$ & $  ( p_{3} - p_{5} )   \cdot  v_{12}  = (39,-30,-10) \cdot (4,0,-32) = 476$ \\ 
 $5  3  4  1  2 $ & $f_{12}$  & $p_{5}$ & $p_{3}$ & $  ( p_{3} - p_{5} )   \cdot  v_{12}  = (39,-30,-10) \cdot (4,0,-32) = 476$ \\ 
 $5  3  2  4  1 $ & $f_{12}$  & $p_{5}$ & $p_{3}$ & $  ( p_{3} - p_{5} )   \cdot  v_{12}  = (39,-30,-10) \cdot (4,0,-32) = 476$ \\ 
 $5  3  2  1  4 $ & $f_{12}$  & $p_{5}$ & $p_{3}$ & $  ( p_{3} - p_{5} )   \cdot  v_{12}  = (39,-30,-10) \cdot (4,0,-32) = 476$ \\ 
 $5  3  1  2  4 $ & $f_{12}$  & $p_{5}$ & $p_{3}$ & $  ( p_{3} - p_{5} )   \cdot  v_{12}  = (39,-30,-10) \cdot (4,0,-32) = 476$ \\ 
 $5  3  1  4  2 $ & $f_{12}$  & $p_{5}$ & $p_{3}$ & $  ( p_{3} - p_{5} )   \cdot  v_{12}  = (39,-30,-10) \cdot (4,0,-32) = 476$ \\ 
 $5  2  3  4  1 $ & $f_{12}$  & $p_{5}$ & $p_{2}$ & $  ( p_{2} - p_{5} )   \cdot  v_{12}  = (40,-21,0) \cdot (4,0,-32) = 160$ \\ 
 $5  2  3  1  4 $ & $f_{12}$  & $p_{5}$ & $p_{2}$ & $  ( p_{2} - p_{5} )   \cdot  v_{12}  = (40,-21,0) \cdot (4,0,-32) = 160$ \\ 
 $5  2  4  3  1 $ & $f_{12}$  & $p_{5}$ & $p_{2}$ & $  ( p_{2} - p_{5} )   \cdot  v_{12}  = (40,-21,0) \cdot (4,0,-32) = 160$ \\ 
 $5  2  4  1  3 $ & $f_{12}$  & $p_{5}$ & $p_{2}$ & $  ( p_{2} - p_{5} )   \cdot  v_{12}  = (40,-21,0) \cdot (4,0,-32) = 160$ \\ 
 $5  2  1  4  3 $ & $f_{12}$  & $p_{5}$ & $p_{2}$ & $  ( p_{2} - p_{5} )   \cdot  v_{12}  = (40,-21,0) \cdot (4,0,-32) = 160$ \\ 
 $5  2  1  3  4 $ & $f_{12}$  & $p_{5}$ & $p_{2}$ & $  ( p_{2} - p_{5} )   \cdot  v_{12}  = (40,-21,0) \cdot (4,0,-32) = 160$ \\ 
 $5  1  3  2  4 $ & $f_{12}$  & $p_{5}$ & $p_{1}$ & $  ( p_{1} - p_{5} )   \cdot  v_{12}  = (14,55,-3) \cdot (4,0,-32) = 152$ \\ 
 $5  1  3  4  2 $ & $f_{12}$  & $p_{5}$ & $p_{1}$ & $  ( p_{1} - p_{5} )   \cdot  v_{12}  = (14,55,-3) \cdot (4,0,-32) = 152$ \\ 
 $5  1  2  3  4 $ & $f_{12}$  & $p_{5}$ & $p_{1}$ & $  ( p_{1} - p_{5} )   \cdot  v_{12}  = (14,55,-3) \cdot (4,0,-32) = 152$ \\ 
 $5  1  2  4  3 $ & $f_{12}$  & $p_{5}$ & $p_{1}$ & $  ( p_{1} - p_{5} )   \cdot  v_{12}  = (14,55,-3) \cdot (4,0,-32) = 152$ \\ 
 $5  1  4  2  3 $ & $f_{12}$  & $p_{5}$ & $p_{1}$ & $  ( p_{1} - p_{5} )   \cdot  v_{12}  = (14,55,-3) \cdot (4,0,-32) = 152$ \\ 
 $5  1  4  3  2 $ & $f_{12}$  & $p_{5}$ & $p_{1}$ & $  ( p_{1} - p_{5} )   \cdot  v_{12}  = (14,55,-3) \cdot (4,0,-32) = 152$ \\ 
 $4  5  3  2  1 $ & $f_{24}$  & $p_{5}$ & $p_{2}$ & $  ( p_{2} - p_{5} )   \cdot  v_{24}  = (40,-21,0) \cdot (4,0,32) = 160$ \\ 
 $4  5  3  1  2 $ & $f_{25}$  & $p_{5}$ & $p_{2}$ & $  ( p_{2} - p_{5} )   \cdot  v_{25}  = (40,-21,0) \cdot (0,-4,32) = 84$ \\ 
 $4  5  2  3  1 $ & $f_{15}$  & $p_{4}$ & $p_{1}$ & $  ( p_{1} - p_{4} )   \cdot  v_{15}  = (-49,93,21) \cdot (1,-3,16) = 8$ \\ 
 $4  5  2  1  3 $ & $f_{15}$  & $p_{4}$ & $p_{3}$ & $  ( p_{3} - p_{4} )   \cdot  v_{15}  = (-24,8,14) \cdot (1,-3,16) = 176$ \\ 
 $4  5  1  2  3 $ & $f_{23}$  & $p_{5}$ & $p_{1}$ & $  ( p_{1} - p_{5} )   \cdot  v_{23}  = (14,55,-3) \cdot (0,4,32) = 124$ \\ 
 $4  5  1  3  2 $ & $f_{23}$  & $p_{5}$ & $p_{1}$ & $  ( p_{1} - p_{5} )   \cdot  v_{23}  = (14,55,-3) \cdot (0,4,32) = 124$ \\ 
 $4  3  5  2  1 $ & $f_{23}$  & $p_{3}$ & $p_{5}$ & $  ( p_{5} - p_{3} )   \cdot  v_{23}  = (-39,30,10) \cdot (0,4,32) = 440$ \\ 
 $4  3  5  1  2 $ & $f_{23}$  & $p_{3}$ & $p_{5}$ & $  ( p_{5} - p_{3} )   \cdot  v_{23}  = (-39,30,10) \cdot (0,4,32) = 440$ \\ 
 $4  3  2  5  1 $ & $f_{23}$  & $p_{3}$ & $p_{2}$ & $  ( p_{2} - p_{3} )   \cdot  v_{23}  = (1,9,10) \cdot (0,4,32) = 356$ \\ 
 $4  3  2  1  5 $ & $f_{23}$  & $p_{3}$ & $p_{2}$ & $  ( p_{2} - p_{3} )   \cdot  v_{23}  = (1,9,10) \cdot (0,4,32) = 356$ \\ 
 $4  3  1  2  5 $ & $f_{23}$  & $p_{3}$ & $p_{1}$ & $  ( p_{1} - p_{3} )   \cdot  v_{23}  = (-25,85,7) \cdot (0,4,32) = 564$ \\ 
 $4  3  1  5  2 $ & $f_{23}$  & $p_{3}$ & $p_{1}$ & $  ( p_{1} - p_{3} )   \cdot  v_{23}  = (-25,85,7) \cdot (0,4,32) = 564$ \\ 
 $4  2  3  5  1 $ & $f_{15}$  & $p_{4}$ & $p_{1}$ & $  ( p_{1} - p_{4} )   \cdot  v_{15}  = (-49,93,21) \cdot (1,-3,16) = 8$ \\ 
 $4  2  3  1  5 $ & $f_{15}$  & $p_{4}$ & $p_{5}$ & $  ( p_{5} - p_{4} )   \cdot  v_{15}  = (-63,38,24) \cdot (1,-3,16) = 207$ \\ 
 $4  2  5  3  1 $ & $f_{23}$  & $p_{2}$ & $p_{5}$ & $  ( p_{5} - p_{2} )   \cdot  v_{23}  = (-40,21,0) \cdot (0,4,32) = 84$ \\ 
 $4  2  5  1  3 $ & $f_{23}$  & $p_{2}$ & $p_{5}$ & $  ( p_{5} - p_{2} )   \cdot  v_{23}  = (-40,21,0) \cdot (0,4,32) = 84$ \\ 
 $4  2  1  5  3 $ & $f_{23}$  & $p_{2}$ & $p_{1}$ & $  ( p_{1} - p_{2} )   \cdot  v_{23}  = (-26,76,-3) \cdot (0,4,32) = 208$ \\ 
 $4  2  1  3  5 $ & $f_{23}$  & $p_{2}$ & $p_{1}$ & $  ( p_{1} - p_{2} )   \cdot  v_{23}  = (-26,76,-3) \cdot (0,4,32) = 208$ \\ 
 $4  1  3  2  5 $ & $f_{24}$  & $p_{1}$ & $p_{2}$ & $  ( p_{2} - p_{1} )   \cdot  v_{24}  = (26,-76,3) \cdot (4,0,32) = 200$ \\ 
 $4  1  3  5  2 $ & $f_{24}$  & $p_{1}$ & $p_{5}$ & $  ( p_{5} - p_{1} )   \cdot  v_{24}  = (-14,-55,3) \cdot (4,0,32) = 40$ \\ 
 $4  1  2  3  5 $ & $f_{25}$  & $p_{1}$ & $p_{5}$ & $  ( p_{5} - p_{1} )   \cdot  v_{25}  = (-14,-55,3) \cdot (0,-4,32) = 316$ \\ 
 $4  1  2  5  3 $ & $f_{24}$  & $p_{1}$ & $p_{5}$ & $  ( p_{5} - p_{1} )   \cdot  v_{24}  = (-14,-55,3) \cdot (4,0,32) = 40$ \\ 
 $4  1  5  2  3 $ & $f_{24}$  & $p_{1}$ & $p_{2}$ & $  ( p_{2} - p_{1} )   \cdot  v_{24}  = (26,-76,3) \cdot (4,0,32) = 200$ \\ 
 $4  1  5  3  2 $ & $f_{25}$  & $p_{1}$ & $p_{2}$ & $  ( p_{2} - p_{1} )   \cdot  v_{25}  = (26,-76,3) \cdot (0,-4,32) = 400$ \\ 
 $3  4  5  2  1 $ & $f_{12}$  & $p_{3}$ & $p_{4}$ & $  ( p_{4} - p_{3} )   \cdot  v_{12}  = (24,-8,-14) \cdot (4,0,-32) = 544$ \\ 
 $3  4  5  1  2 $ & $f_{12}$  & $p_{3}$ & $p_{4}$ & $  ( p_{4} - p_{3} )   \cdot  v_{12}  = (24,-8,-14) \cdot (4,0,-32) = 544$ \\ 
 $3  4  2  5  1 $ & $f_{12}$  & $p_{3}$ & $p_{4}$ & $  ( p_{4} - p_{3} )   \cdot  v_{12}  = (24,-8,-14) \cdot (4,0,-32) = 544$ \\ 
 $3  4  2  1  5 $ & $f_{12}$  & $p_{3}$ & $p_{4}$ & $  ( p_{4} - p_{3} )   \cdot  v_{12}  = (24,-8,-14) \cdot (4,0,-32) = 544$ \\ 
 $3  4  1  2  5 $ & $f_{12}$  & $p_{3}$ & $p_{4}$ & $  ( p_{4} - p_{3} )   \cdot  v_{12}  = (24,-8,-14) \cdot (4,0,-32) = 544$ \\ 
 $3  4  1  5  2 $ & $f_{12}$  & $p_{3}$ & $p_{4}$ & $  ( p_{4} - p_{3} )   \cdot  v_{12}  = (24,-8,-14) \cdot (4,0,-32) = 544$ \\ 
 $3  5  4  2  1 $ & $f_{24}$  & $p_{5}$ & $p_{2}$ & $  ( p_{2} - p_{5} )   \cdot  v_{24}  = (40,-21,0) \cdot (4,0,32) = 160$ \\ 
 $3  5  4  1  2 $ & $f_{25}$  & $p_{5}$ & $p_{2}$ & $  ( p_{2} - p_{5} )   \cdot  v_{25}  = (40,-21,0) \cdot (0,-4,32) = 84$ \\ 
 $3  5  2  4  1 $ & $f_{13}$  & $p_{3}$ & $p_{2}$ & $  ( p_{2} - p_{3} )   \cdot  v_{13}  = (1,9,10) \cdot (2,2,0) = 20$ \\ 
 $3  5  2  1  4 $ & $f_{13}$  & $p_{3}$ & $p_{2}$ & $  ( p_{2} - p_{3} )   \cdot  v_{13}  = (1,9,10) \cdot (2,2,0) = 20$ \\ 
 $3  5  1  2  4 $ & $f_{23}$  & $p_{5}$ & $p_{1}$ & $  ( p_{1} - p_{5} )   \cdot  v_{23}  = (14,55,-3) \cdot (0,4,32) = 124$ \\ 
 $3  5  1  4  2 $ & $f_{23}$  & $p_{5}$ & $p_{1}$ & $  ( p_{1} - p_{5} )   \cdot  v_{23}  = (14,55,-3) \cdot (0,4,32) = 124$ \\ 
 $3  2  5  4  1 $ & $f_{23}$  & $p_{2}$ & $p_{5}$ & $  ( p_{5} - p_{2} )   \cdot  v_{23}  = (-40,21,0) \cdot (0,4,32) = 84$ \\ 
 $3  2  5  1  4 $ & $f_{23}$  & $p_{2}$ & $p_{5}$ & $  ( p_{5} - p_{2} )   \cdot  v_{23}  = (-40,21,0) \cdot (0,4,32) = 84$ \\ 
 $3  2  4  5  1 $ & $f_{13}$  & $p_{3}$ & $p_{4}$ & $  ( p_{4} - p_{3} )   \cdot  v_{13}  = (24,-8,-14) \cdot (2,2,0) = 32$ \\ 
 $3  2  4  1  5 $ & $f_{15}$  & $p_{3}$ & $p_{5}$ & $  ( p_{5} - p_{3} )   \cdot  v_{15}  = (-39,30,10) \cdot (1,-3,16) = 31$ \\ 
 $3  2  1  4  5 $ & $f_{23}$  & $p_{2}$ & $p_{1}$ & $  ( p_{1} - p_{2} )   \cdot  v_{23}  = (-26,76,-3) \cdot (0,4,32) = 208$ \\ 
 $3  2  1  5  4 $ & $f_{23}$  & $p_{2}$ & $p_{1}$ & $  ( p_{1} - p_{2} )   \cdot  v_{23}  = (-26,76,-3) \cdot (0,4,32) = 208$ \\ 
 $3  1  5  2  4 $ & $f_{24}$  & $p_{1}$ & $p_{2}$ & $  ( p_{2} - p_{1} )   \cdot  v_{24}  = (26,-76,3) \cdot (4,0,32) = 200$ \\ 
 $3  1  5  4  2 $ & $f_{25}$  & $p_{1}$ & $p_{2}$ & $  ( p_{2} - p_{1} )   \cdot  v_{25}  = (26,-76,3) \cdot (0,-4,32) = 400$ \\ 
 $3  1  2  5  4 $ & $f_{24}$  & $p_{1}$ & $p_{5}$ & $  ( p_{5} - p_{1} )   \cdot  v_{24}  = (-14,-55,3) \cdot (4,0,32) = 40$ \\ 
 $3  1  2  4  5 $ & $f_{25}$  & $p_{1}$ & $p_{5}$ & $  ( p_{5} - p_{1} )   \cdot  v_{25}  = (-14,-55,3) \cdot (0,-4,32) = 316$ \\ 
 $3  1  4  2  5 $ & $f_{24}$  & $p_{1}$ & $p_{2}$ & $  ( p_{2} - p_{1} )   \cdot  v_{24}  = (26,-76,3) \cdot (4,0,32) = 200$ \\ 
 $3  1  4  5  2 $ & $f_{24}$  & $p_{1}$ & $p_{5}$ & $  ( p_{5} - p_{1} )   \cdot  v_{24}  = (-14,-55,3) \cdot (4,0,32) = 40$ \\ 
 $2  4  3  5  1 $ & $f_{12}$  & $p_{2}$ & $p_{4}$ & $  ( p_{4} - p_{2} )   \cdot  v_{12}  = (23,-17,-24) \cdot (4,0,-32) = 860$ \\ 
 $2  4  3  1  5 $ & $f_{12}$  & $p_{2}$ & $p_{4}$ & $  ( p_{4} - p_{2} )   \cdot  v_{12}  = (23,-17,-24) \cdot (4,0,-32) = 860$ \\ 
 $2  4  5  3  1 $ & $f_{12}$  & $p_{2}$ & $p_{4}$ & $  ( p_{4} - p_{2} )   \cdot  v_{12}  = (23,-17,-24) \cdot (4,0,-32) = 860$ \\ 
 $2  4  5  1  3 $ & $f_{12}$  & $p_{2}$ & $p_{4}$ & $  ( p_{4} - p_{2} )   \cdot  v_{12}  = (23,-17,-24) \cdot (4,0,-32) = 860$ \\ 
 $2  4  1  5  3 $ & $f_{12}$  & $p_{2}$ & $p_{4}$ & $  ( p_{4} - p_{2} )   \cdot  v_{12}  = (23,-17,-24) \cdot (4,0,-32) = 860$ \\ 
 $2  4  1  3  5 $ & $f_{12}$  & $p_{2}$ & $p_{4}$ & $  ( p_{4} - p_{2} )   \cdot  v_{12}  = (23,-17,-24) \cdot (4,0,-32) = 860$ \\ 
 $2  3  4  5  1 $ & $f_{12}$  & $p_{2}$ & $p_{3}$ & $  ( p_{3} - p_{2} )   \cdot  v_{12}  = (-1,-9,-10) \cdot (4,0,-32) = 316$ \\ 
 $2  3  4  1  5 $ & $f_{12}$  & $p_{2}$ & $p_{3}$ & $  ( p_{3} - p_{2} )   \cdot  v_{12}  = (-1,-9,-10) \cdot (4,0,-32) = 316$ \\ 
 $2  3  5  4  1 $ & $f_{12}$  & $p_{2}$ & $p_{3}$ & $  ( p_{3} - p_{2} )   \cdot  v_{12}  = (-1,-9,-10) \cdot (4,0,-32) = 316$ \\ 
 $2  3  5  1  4 $ & $f_{12}$  & $p_{2}$ & $p_{3}$ & $  ( p_{3} - p_{2} )   \cdot  v_{12}  = (-1,-9,-10) \cdot (4,0,-32) = 316$ \\ 
 $2  3  1  5  4 $ & $f_{12}$  & $p_{2}$ & $p_{3}$ & $  ( p_{3} - p_{2} )   \cdot  v_{12}  = (-1,-9,-10) \cdot (4,0,-32) = 316$ \\ 
 $2  3  1  4  5 $ & $f_{12}$  & $p_{2}$ & $p_{3}$ & $  ( p_{3} - p_{2} )   \cdot  v_{12}  = (-1,-9,-10) \cdot (4,0,-32) = 316$ \\ 
 $2  5  3  4  1 $ & $f_{34}$  & $p_{3}$ & $p_{4}$ & $  ( p_{4} - p_{3} )   \cdot  v_{34}  = (24,-8,-14) \cdot (2,-2,0) = 64$ \\ 
 $2  5  3  1  4 $ & $f_{45}$  & $p_{1}$ & $p_{4}$ & $  ( p_{4} - p_{1} )   \cdot  v_{45}  = (49,-93,-21) \cdot (-2,-3,8) = 13$ \\ 
 $2  5  4  3  1 $ & $f_{13}$  & $p_{2}$ & $p_{4}$ & $  ( p_{4} - p_{2} )   \cdot  v_{13}  = (23,-17,-24) \cdot (2,2,0) = 12$ \\ 
 $2  5  4  1  3 $ & $f_{13}$  & $p_{2}$ & $p_{4}$ & $  ( p_{4} - p_{2} )   \cdot  v_{13}  = (23,-17,-24) \cdot (2,2,0) = 12$ \\ 
 $2  5  1  4  3 $ & $f_{23}$  & $p_{5}$ & $p_{1}$ & $  ( p_{1} - p_{5} )   \cdot  v_{23}  = (14,55,-3) \cdot (0,4,32) = 124$ \\ 
 $2  5  1  3  4 $ & $f_{23}$  & $p_{5}$ & $p_{1}$ & $  ( p_{1} - p_{5} )   \cdot  v_{23}  = (14,55,-3) \cdot (0,4,32) = 124$ \\ 
 $2  1  3  5  4 $ & $f_{24}$  & $p_{1}$ & $p_{5}$ & $  ( p_{5} - p_{1} )   \cdot  v_{24}  = (-14,-55,3) \cdot (4,0,32) = 40$ \\ 
 $2  1  3  4  5 $ & $f_{25}$  & $p_{1}$ & $p_{5}$ & $  ( p_{5} - p_{1} )   \cdot  v_{25}  = (-14,-55,3) \cdot (0,-4,32) = 316$ \\ 
 $2  1  5  3  4 $ & $f_{34}$  & $p_{5}$ & $p_{3}$ & $  ( p_{3} - p_{5} )   \cdot  v_{34}  = (39,-30,-10) \cdot (2,-2,0) = 138$ \\ 
 $2  1  5  4  3 $ & $f_{25}$  & $p_{1}$ & $p_{3}$ & $  ( p_{3} - p_{1} )   \cdot  v_{25}  = (25,-85,-7) \cdot (0,-4,32) = 116$ \\ 
 $2  1  4  5  3 $ & $f_{24}$  & $p_{1}$ & $p_{5}$ & $  ( p_{5} - p_{1} )   \cdot  v_{24}  = (-14,-55,3) \cdot (4,0,32) = 40$ \\ 
 $2  1  4  3  5 $ & $f_{25}$  & $p_{1}$ & $p_{5}$ & $  ( p_{5} - p_{1} )   \cdot  v_{25}  = (-14,-55,3) \cdot (0,-4,32) = 316$ \\ 
 $1  4  3  2  5 $ & $f_{12}$  & $p_{1}$ & $p_{4}$ & $  ( p_{4} - p_{1} )   \cdot  v_{12}  = (49,-93,-21) \cdot (4,0,-32) = 868$ \\ 
 $1  4  3  5  2 $ & $f_{12}$  & $p_{1}$ & $p_{4}$ & $  ( p_{4} - p_{1} )   \cdot  v_{12}  = (49,-93,-21) \cdot (4,0,-32) = 868$ \\ 
 $1  4  2  3  5 $ & $f_{12}$  & $p_{1}$ & $p_{4}$ & $  ( p_{4} - p_{1} )   \cdot  v_{12}  = (49,-93,-21) \cdot (4,0,-32) = 868$ \\ 
 $1  4  2  5  3 $ & $f_{12}$  & $p_{1}$ & $p_{4}$ & $  ( p_{4} - p_{1} )   \cdot  v_{12}  = (49,-93,-21) \cdot (4,0,-32) = 868$ \\ 
 $1  4  5  2  3 $ & $f_{12}$  & $p_{1}$ & $p_{4}$ & $  ( p_{4} - p_{1} )   \cdot  v_{12}  = (49,-93,-21) \cdot (4,0,-32) = 868$ \\ 
 $1  4  5  3  2 $ & $f_{12}$  & $p_{1}$ & $p_{4}$ & $  ( p_{4} - p_{1} )   \cdot  v_{12}  = (49,-93,-21) \cdot (4,0,-32) = 868$ \\ 
 $1  3  4  2  5 $ & $f_{12}$  & $p_{1}$ & $p_{3}$ & $  ( p_{3} - p_{1} )   \cdot  v_{12}  = (25,-85,-7) \cdot (4,0,-32) = 324$ \\ 
 $1  3  4  5  2 $ & $f_{12}$  & $p_{1}$ & $p_{3}$ & $  ( p_{3} - p_{1} )   \cdot  v_{12}  = (25,-85,-7) \cdot (4,0,-32) = 324$ \\ 
 $1  3  2  4  5 $ & $f_{12}$  & $p_{1}$ & $p_{3}$ & $  ( p_{3} - p_{1} )   \cdot  v_{12}  = (25,-85,-7) \cdot (4,0,-32) = 324$ \\ 
 $1  3  2  5  4 $ & $f_{12}$  & $p_{1}$ & $p_{3}$ & $  ( p_{3} - p_{1} )   \cdot  v_{12}  = (25,-85,-7) \cdot (4,0,-32) = 324$ \\ 
 $1  3  5  2  4 $ & $f_{12}$  & $p_{1}$ & $p_{3}$ & $  ( p_{3} - p_{1} )   \cdot  v_{12}  = (25,-85,-7) \cdot (4,0,-32) = 324$ \\ 
 $1  3  5  4  2 $ & $f_{12}$  & $p_{1}$ & $p_{3}$ & $  ( p_{3} - p_{1} )   \cdot  v_{12}  = (25,-85,-7) \cdot (4,0,-32) = 324$ \\ 
 $1  2  3  4  5 $ & $f_{12}$  & $p_{1}$ & $p_{2}$ & $  ( p_{2} - p_{1} )   \cdot  v_{12}  = (26,-76,3) \cdot (4,0,-32) = 8$ \\ 
 $1  2  3  5  4 $ & $f_{12}$  & $p_{1}$ & $p_{2}$ & $  ( p_{2} - p_{1} )   \cdot  v_{12}  = (26,-76,3) \cdot (4,0,-32) = 8$ \\ 
 $1  2  4  3  5 $ & $f_{12}$  & $p_{1}$ & $p_{2}$ & $  ( p_{2} - p_{1} )   \cdot  v_{12}  = (26,-76,3) \cdot (4,0,-32) = 8$ \\ 
 $1  2  4  5  3 $ & $f_{12}$  & $p_{1}$ & $p_{2}$ & $  ( p_{2} - p_{1} )   \cdot  v_{12}  = (26,-76,3) \cdot (4,0,-32) = 8$ \\ 
 $1  2  5  4  3 $ & $f_{12}$  & $p_{1}$ & $p_{2}$ & $  ( p_{2} - p_{1} )   \cdot  v_{12}  = (26,-76,3) \cdot (4,0,-32) = 8$ \\ 
 $1  2  5  3  4 $ & $f_{12}$  & $p_{1}$ & $p_{2}$ & $  ( p_{2} - p_{1} )   \cdot  v_{12}  = (26,-76,3) \cdot (4,0,-32) = 8$ \\ 
 $1  5  3  2  4 $ & $f_{24}$  & $p_{5}$ & $p_{2}$ & $  ( p_{2} - p_{5} )   \cdot  v_{24}  = (40,-21,0) \cdot (4,0,32) = 160$ \\ 
 $1  5  3  4  2 $ & $f_{25}$  & $p_{5}$ & $p_{2}$ & $  ( p_{2} - p_{5} )   \cdot  v_{25}  = (40,-21,0) \cdot (0,-4,32) = 84$ \\ 
 $1  5  2  3  4 $ & $f_{34}$  & $p_{2}$ & $p_{3}$ & $  ( p_{3} - p_{2} )   \cdot  v_{34}  = (-1,-9,-10) \cdot (2,-2,0) = 16$ \\ 
 $1  5  2  4  3 $ & $f_{15}$  & $p_{1}$ & $p_{3}$ & $  ( p_{3} - p_{1} )   \cdot  v_{15}  = (25,-85,-7) \cdot (1,-3,16) = 168$ \\ 
 $1  5  4  2  3 $ & $f_{24}$  & $p_{5}$ & $p_{2}$ & $  ( p_{2} - p_{5} )   \cdot  v_{24}  = (40,-21,0) \cdot (4,0,32) = 160$ \\ 
 $1  5  4  3  2 $ & $f_{25}$  & $p_{5}$ & $p_{2}$ & $  ( p_{2} - p_{5} )   \cdot  v_{25}  = (40,-21,0) \cdot (0,-4,32) = 84$ \\ 
\end{longtable}

We show now that the polyhedral fan does not contain a cycle in any direction. 
%
%
%
To this end, we check all simple cycles in the dual graph, which has one vertex per cell and one edge per facet. 
If the halfspaces defined by a cycle (with the corresponding orientation) have empty intersection, then there is no cycle in visibility involving the facets in the cycle. 
The graph associated to $\mathcal{F}$ is depicted in~\cref{fig:anrrgraph}. 

\begin{figure}[htbp]
\centering
\includegraphics[page=3,scale=1]{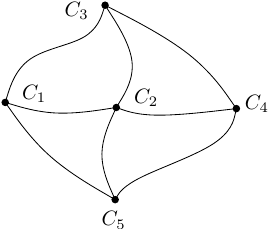}
\caption{Dual graph of $\mathcal{F}$.}
\label{fig:anrrgraph}
\end{figure}

Note first that the halfspaces associated to cycles around a ray have empty intersection. 
This excludes the cycles $ \langle  C_1,C_2,C_3 \rangle$,$ \langle  C_2,C_3,C_4 \rangle$, $ \langle  C_2,C_4,C_5 \rangle$ and $ \langle  C_1,C_2,C_5 \rangle$.

We give next certificates for the remaining cycles.
To show that the halfspaces have empty intersection we provide a solution to the dual problem (in the sense of~\cref{Gordan}). 

\begin{align*}
\langle  C_1,C_3,C_4,C_5 \rangle &: v_{13}+\frac{1}{4} v_{34}+v_{45}-\frac{1}{2} v_{15}=0\\
\langle  C_2,C_5,C_1,C_3 \rangle &: \frac{3}{4}v_{25}-\frac{1}{4}v_{15}+\frac{1}{2}v_{13}-\frac{1}{4} v_{23}=0\\
\langle  C_2,C_1,C_3,C_4 \rangle &:  -\frac{1}{2}v_{12}+v_{13}+v_{34}-\frac{1}{2} v_{24}=0\\
\langle  C_2,C_3,C_4,C_5 \rangle &: \frac{1}{2}v_{23} +v_{34}+v_{45}-\frac{3}{4} v_{25}=0\\
\langle  C_2,C_4,C_5,C_1 \rangle &: \frac{1}{2}v_{24}+v_{45}-v_{15}+\frac{1}{4} v_{12}=0 \\
\langle  C_2,C_1,C_3,C_4,C_5 \rangle &: -\frac{1}{2}v_{12}+v_{13}+v_{34}+v_{45}-\frac{3}{4} v_{25}=0\\
\langle  C_2,C_3,C_4,C_5,C_1 \rangle &: \frac{1}{2}v_{23}+v_{34}+v_{45}-v_{15}+\frac{1}{4} v_{12}=0\\
\langle  C_2,C_4,C_5,C_1,C_3 \rangle &: \frac{1}{2}v_{24}+v_{45}-v_{15}+\frac{1}{2} v_{13}-\frac{1}{4} v_{23}=0\\
\langle  C_2,C_5,C_1,C_3,C_4 \rangle &: \frac{3}{4}v_{25}-\frac{1}{4}v_{15}+\frac{1}{2}v_{13}+\frac{1}{4} v_{34}-\frac{1}{4} v_{24}=0
\end{align*}

\section[A non-recursively-regular acyclic triangulation]{A non-recursively-regular acyclic triangulation}\label{app:notrecreg}

We prove here that the triangulation $\mathcal{S}$ pictured in~\cref{Fig:doublegen} is not recursively-regular. 
To this end, we prove that its finest regular coarsening has only one cell providing a solution to the dual of its regularity system in the sense of~\cref{Gordan}. 

\begin{figure}[htbp]
\centering
\includegraphics[page=5,scale=1]{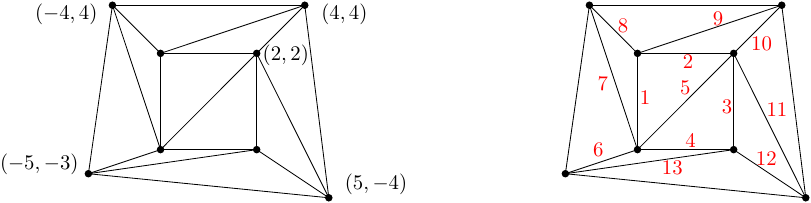}
\caption{A non-recursively-regular triangulation.}
\label{fig:nonrecreeg}
\end{figure}

The rows of the matrix of the regularity system of $\mathcal{S}$ associated to the edges labeled in~\cref{fig:nonrecreeg} are
\begin{align*}
s_1 & =(8,-32,8,0,16,0,0,0) \\
s_2 & =(8,0,-24,0,0,16,0,0) \\
s_3 & =(12,0,8,-36,0,0,16,0) \\
s_4 & =(-28,0,4,8,0,0,0,16) \\
s_5 & =(-16,16,-16,16,0,0,0,0) \\
s_6 & =(-48,0,0,20,4,0,0,24) \\
s_7 & =(-16,20,0,0,-12,0,0,8) \\
s_8 & =(16,-48,0,0,24,8,0,0) \\
s_9 & =(0,-16,16,0,8,-8,0,0) \\
s_{10} & =(0,18,-50,0,0,24,8,0) \\
s_{11} & =(0,0,-22,18,0,12,-8,0) \\
s_{12} & =(0,0,17,-57,0,0,28,12) \\
s_{13} & =(17,0,0,-13,0,0,4,-8).
\end{align*}

The following positive values for the dual variables are a solution to the dual system 
\[
\sum_{i=1}^{13} y_i s_i=0\]
\[ y_i \ge 0, \text{ for } i=1,\ldots,13\]



%

\[y_1=207,y_2=24,y_3=20,y_4=24,y_5=288,y_6=24,y_7=1308\]
\[y_8=24,y_9=1464,y_{10}=24, y_{11}=198,y_{12}=24,y_{13}=1464.\]

\bigskip

We prove now that $\mathcal{S}$ has no cycle in visibility. 
We show first that the coarsening $\mathcal{S}'$ of $\mathcal{S}$ depicted in~\cref{fig:nonrecreegcor} is acyclic. 
Assume $x \in \R^2$ to be a point interior to a cell $C_x$ of $\mathcal{S}'$ from which $\mathcal{S}'$ might be cyclic.
Consider the fan $\mathcal{F} \subset \R^3$ obtained by taking the cone with the origin for all the cells of $\mathcal{S}'$, when $\mathcal{S}'$ is embedded in the horizontal plane of intercept $-1/8$ with the vertical axis.
The construction is similar to the fan studied in~\cref{App:calculations}. 
In particular, both subdivisions have the same dual graph (illustrated in \cref{fig:anrrgraph}). 

We will prove that the in-front relation from $x$ in the plane for $\mathcal{S}'$ is equivalent to the in-front relation in the direction $(-x, 1/8) \in \R^3$ for $\mathcal{F}$, which is shown to be acyclic in~\cref{App:calculations}.
Therefore, it will be proven that $\mathcal{S}'$ is acyclic as well.
Take a line $\ell$ in the plane through $x$ and consider points $y,z \in \ell$ such that $y$ is interior to $C_y$ and $z$ is interior to $C_z$ with $C_y,C_z \in \cells(\mathcal{S}')$ sharing the wall $W=C_y \cap C_z$ of $\mathcal{S}'$. 
Assume further that $y$ is before $z$ in the visibility relation from $x$. 
Let $\Pi \subset \R^3$ be a plane through the origin and containing $W$ in the embedded copy of $\mathcal{S}'$ (and, hence, supporting the wall of $\mathcal{F}$ associated to $W$). 
Then, the plane~$\Pi$ separates $\lambda y$ from $\mu z$ for any $\lambda, \mu \in \R^+$.
In particular, it separates $\lambda(\mu) y$ from $\mu z$ if for every $\mu \in \R^+$ we take $\lambda(\mu)$ such that $x$, $\lambda(\mu) y$ and $\mu z$ are coplanar. 
If $v$ is a vector normal to $\Pi$ and pointing ``towards'' $x$, we have then that $\scalprod{v}{\lambda(\mu)y-\mu z}>0$ for all $\mu \in \R^+$. 
Since the scalar product is continuous, we have that 
\[ \lim_{ \mu \to 0} \scalprod{v}{\lambda(\mu)y-\mu z} = \scalprod{v}{ \alpha x} \ge 0 ,\]
where $\alpha \in \R^+$ is a constant. 
Since $\scalprod{v}{x} \ne 0$ because $x$ is interior to $C_x$, we have that $\scalprod{v}{x} > 0$.

It only remains to observe that all the sub-subdivisions refining each of the cells of $\mathcal{S}'$ into a pair of cells of $\mathcal{S}$ are obviously acyclic. 
By a reasoning similar to the one in the proof of~\cref{recareacyclic}, the combination of the previous observation with the fact that $\mathcal{S}'$ is acyclic already implies that $\mathcal{S}$ is acyclic.

\begin{figure}[htbp]
\centering
\includegraphics[page=6,scale=0.7]{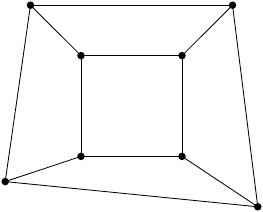}
\caption{An acyclic coarsening of $\mathcal{S}$.}
\label{fig:nonrecreegcor}
\end{figure}

\section[A non-regular recursively-regular subdivision]{A non-regular recursively-regular subdivision}\label{Ap:nonreg}

We prove here that the subdivision $\mathcal{S}$ pictured in~\cref{Fig:doublegen} is not regular. 
To this end, we provide a solution to the dual of its regularity system in the sense of~\cref{Gordan}. 
The vertices of the subdivision lie symmetrically with respect to the coordinate axes. 

\begin{figure}[htbp]
\centering
\includegraphics[scale=0.7]{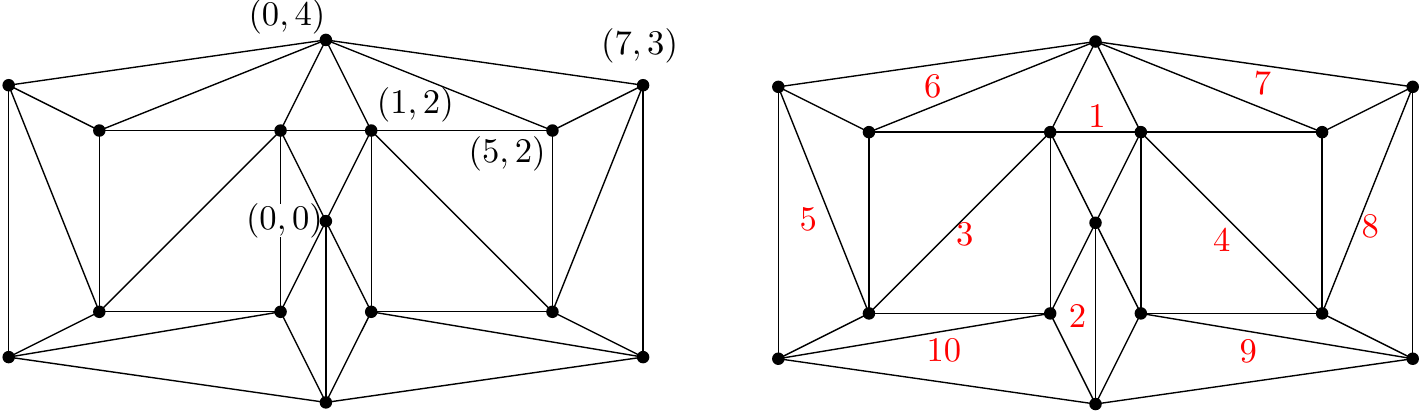}
\caption{A non-regular, recursively regular subdivision.}
\label{fig:nonreg}
\end{figure}

The rows of the matrix of the regularity system of $\mathcal{S}$ associated to the edges labeled in~\cref{fig:nonreg} are
\begin{align*}
s_1 & =(0,0,-4,0,4,-4,0,0,0,4,0,0,0,0,0) \\
s_2 & =(0,0,0,4,-4,0,4,0,0,0,-4,0,0,0,0) \\
s_3 & =(-16,16,-16,16,0,0,0,0,0,0,0,0,0,0,0) \\
s_4 & =(0,0,0,0,0,-16,16,-16,16,0,0,0,0,0,0) \\
s_5 & =(-12,12,0,0,0,0,0,0,0,0,0,8,0,0,-8) \\
s_6 & =(0,-13,9,0,0,0,0,0,0,-4,0,0,0,0,8) \\
s_7 & =(0,0,0,0,0,9,0,0,-13,-4,0,0,0,8,0) \\
s_8 & =(0,0,0,0,0,0,0,-12,12,0,0,0,8,-8,0) \\
s_9 & =(0,0,0,0,0,0,0,-9,13,0,0,4,0,-	8,0) \\
s_{10} & =(13,0,0,-9,0,0,0,0,0,0,4,-8,0,0,0).
\end{align*}

The following positive values for the dual variables are a solution to the dual system

\[
\sum_{i=1}^{10} y_i s_i=0\]
\[ y_i \ge 0, \text{ for } i= 1,\ldots,10\]

\[y_{1}=y_{2} =1,y_{3}=y_{4} =\frac{1}{32},y_{5} =y_{6} =y_{7} =y_{8} =y_{9} =y_{10} =\frac{1}{2}.\]
%

The coarsening $\mathcal{S}'$ of $\mathcal{S}$ in~\cref{fig:reccoars} is regular, because the depicted heights represent a convex lifting. 
The subdivisions $\mathcal{S}$ is then recursively-regular because it can be obtained refining with regular subdivisions the cells of $\mathcal{S}'$.

\begin{figure}[htbp]
\centering
\includegraphics[page=2,scale=1]{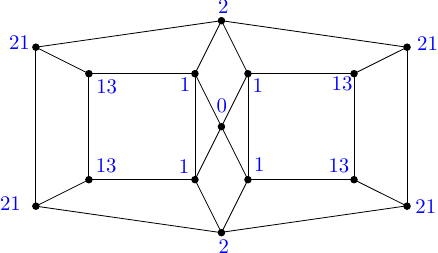}
\caption{A regular coarsening of $\mathcal{S}$.}
\label{fig:reccoars}
\end{figure}

\section[A subdivision whose regularity tree has two levels]{A subdivision whose regularity tree has two levels}
\label{Ap:twolevels}

We prove here that the subdivision $\mathcal{S}$ pictured in~\cref{Fig:rectree} (left) is not regular and that its finest regular coarsening is the subdivision $\mathcal{S}_0$ defined by the second level of the tree in~\cref{Fig:rectree} (right). 
To this end, we provide a solution to the dual of its regularity system in the sense of~\cref{Gordan}. 

\begin{figure}[htbp]
\centering
\includegraphics[scale=0.7]{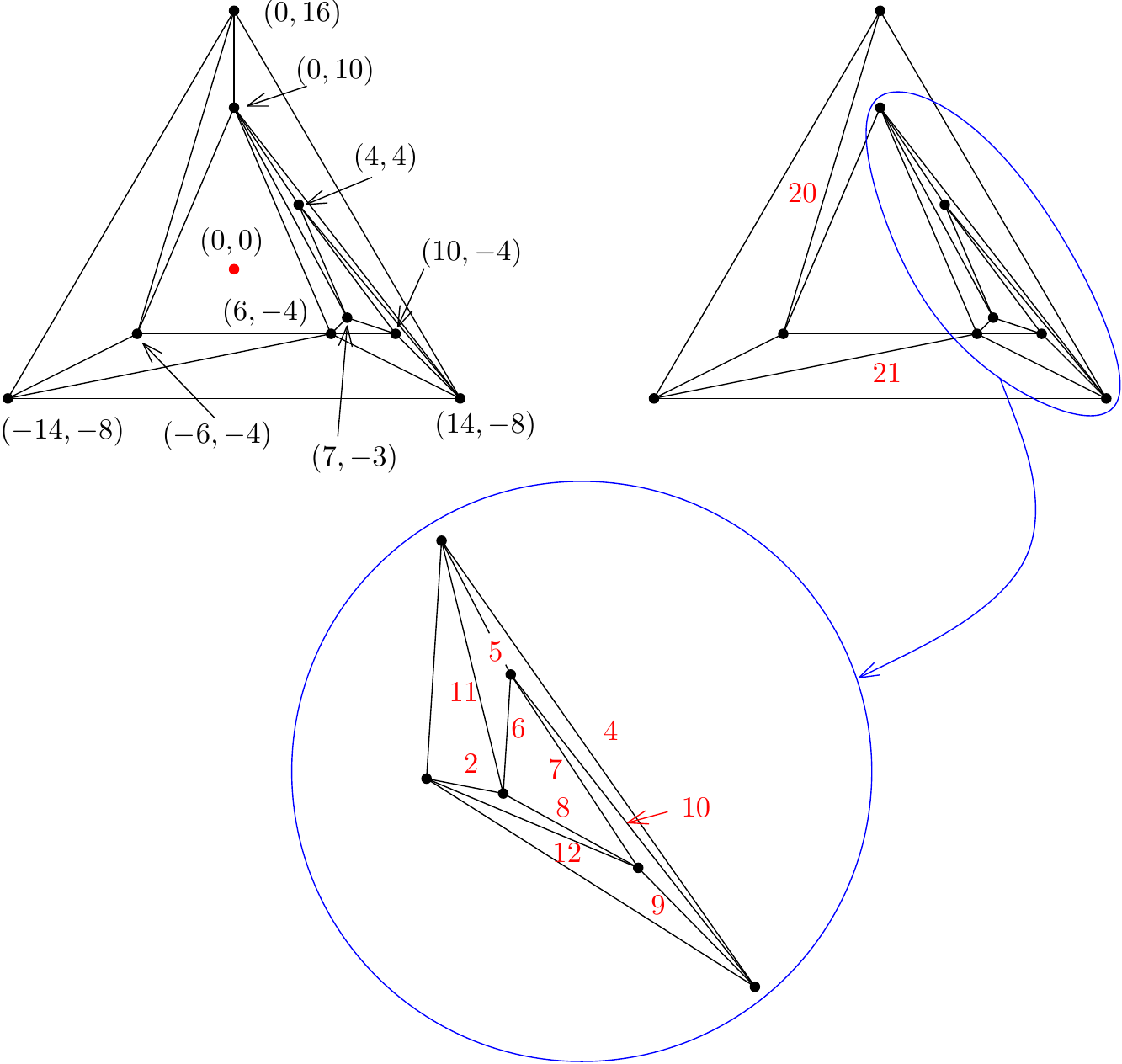}
\caption{A non-regular, recursively regular subdivision.}
\label{fig:twolev}
\end{figure}

The rows of the matrix of the regularity system of $\mathcal{S}$ associated to the edges labeled with numbers in~\cref{fig:twolev} are
\begin{align*}
s_{2} & = (-56,20,0,4,32,0,0,0,0) \\
s_{4} & = (0,0,84,-72,0,-24,0,0,12) \\
s_{5} & = (12,0,-56,34,0,10,0,0,0) \\
s_{6} & = (4,10,-32,18,0,0,0,0,0) \\
s_{7} & = (8,-34,8,0,0,18,0,0,0) \\
s_{8} & = (-32,10,4,0,18,0,0,0,0) \\
s_{9} & = (0,56,16,0,8,32,0,0,0) \\
s_{10} & = (0,12,-16,8,0,-4,0,0,0) \\
s_{11} & = (-20,0,20,-10,10,0,0,0,0) \\
s_{12} & = (16,-12,0,-8,4,0,0,0,0) \\
s_{20} & = (0,0,0,136,0,0,36,-84,-88) \\
s_{21} & = (0,0,0,0,-112,48,-48,112,0).
\end{align*}

The following positive values for the dual variables are a solution to the dual system
\[
\sum_{i=1}^{21} y_i s_i=0\]
\[ y_i \ge 0, \text{ for } i =1,\ldots,21\]
%

\[y_{2}=\frac{1}{10},y_{4}=\frac{11}{10},y_{6}=1,y_{7}=\frac{50}{99},y_{8}=\frac{71}{99}\]
\[y_{9}=\frac{1}{10},y_{10}=\frac{11}{10}, y_{12}=\frac{4}{5}, y_{5}=\frac{109}{110}, y_{11}=\frac{23}{110},y_{20}=\frac{3}{20},y_{21}=\frac{9}{80}.\]

The sub-subdivision of $\mathcal{S}$ pictured (after an affine transformation) in the lower part of~\cref{fig:twolev} is a variant of a typical non-regular subdivision called ``the mother of all examples''in~\cite{LRS}. 
However, since the lines supporting the edges $2$, $5$ and $9$ are concurrent, the sub-subdivision is recursively regular (its finest regular coarsening is the projection of a truncated pyramid). 
Similarly, it is clear that $\mathcal{S}_0$ is regular as well.

\end{document}